\documentclass[a4paper]{article}
\usepackage{geometry}
\usepackage{amsmath}
\usepackage{amssymb}
\usepackage{indentfirst}
\usepackage{mathdots}
\usepackage{cite}
\usepackage[colorlinks, linkcolor=blue, citecolor=blue]{hyperref}
\usepackage{booktabs}
\usepackage{enumitem}
\usepackage{mathrsfs}
\usepackage{enumitem}
\usepackage{authblk}
\usepackage{amsfonts,ntheorem}
\usepackage{float}
\usepackage{setspace}
\numberwithin{equation}{section}
\usepackage{color}

\usepackage[justification=centering]{caption}

\usepackage{graphicx}
\usepackage{subfigure}

\makeatletter

\newcommand{\Rmnum}[1]{\expandafter\@slowromancap\romannumeral #1@}
\makeatother

\newtheorem{theorem}{Theorem}[section]
\newtheorem{lemma}[theorem]{Lemma}

\newtheorem{proposition}[theorem]{Proposition}

{ 
\theoremstyle{plain}
\theorembodyfont{\normalfont} 

\newtheorem{remark}[theorem]{Remark}

\newtheorem{definition}[theorem]{Definition}

}

\newenvironment{proof}{\noindent{\textbf{\emph{Proof.}}}}

\allowdisplaybreaks[4]

\begin{document}

\title{Construction of $\varepsilon_{d}$-ASIC-POVMs via $2$-to-$1$ PN functions and the Li bound}

\author[1]{{\normalsize{Meng Cao}} \thanks{E-mail address: mengcao@bimsa.cn}}
\author[2]{{\normalsize{Xiantao Deng}} \thanks{E-mail address: xiantaodeng@126.com}}
\affil[1]{\footnotesize{Yanqi Lake Beijing Institute of Mathematical Sciences and Applications, Beijing, 101408, China }}
\affil[2]{\footnotesize{Department of Mathematics, Sichuan University, Chengdu, 610064, China }}
\renewcommand*{\Affilfont}{\small\it}

\date{}

\maketitle

\vspace{-15pt}

{\linespread{1.4}{

\begin{abstract}
Symmetric informationally complete positive operator-valued measures (SIC-POVMs) in finite dimension $d$ are a particularly attractive case of informationally complete POVMs (IC-POVMs), which consist of $d^{2}$ subnormalized projectors with equal pairwise fidelity. However, it is difficult to construct SIC-POVMs, and it is not even clear whether there exists an infinite family of SIC-POVMs. To realize some possible applications in quantum information processing, Klappenecker \emph{et al.} \cite{Klappenecker2005On} introduced an approximate version of SIC-POVMs called approximately symmetric informationally complete POVMs (ASIC-POVMs).
In this paper, we construct a class of $\varepsilon_{d}$-ASIC-POVMs in dimension $d=q$ and a class of $\varepsilon_{d}$-ASIC-POVMs in dimension $d=q+1$, respectively,
where $q$ is a prime power.
We prove that all $2$-to-$1$ perfect nonlinear (PN) functions can be used for constructing $\varepsilon_{q}$-ASIC-POVMs.
We show that the set of vectors corresponding to the $\varepsilon_{q}$-ASIC-POVM forms a biangular frame.
The construction of $\varepsilon_{q+1}$-ASIC-POVMs is based on a multiplicative character sum estimate called the Li bound.
We show that the set of vectors corresponding to the $\varepsilon_{q+1}$-ASIC-POVM forms an asymptotically optimal codebook.
We characterize ``how close'' the $\varepsilon_{q}$-ASIC-POVMs (resp. $\varepsilon_{q+1}$-ASIC-POVMs) are from being
SIC-POVMs of dimension $q$ (resp. dimension $q+1$). Finally, we explain the significance of constructing $\varepsilon_{d}$-ASIC-POVMs.
\end{abstract}


\section{Introduction}
Suppose $d$ is a positive integer. Let $\mathbb{C}^{d}$ denote the $d$-dimensional vector space over the complex field $\mathbb{C}$.
Denote by $M_{d}(\mathbb{C})$ the set consisting of all $d\times d$ matrices with entries in $\mathbb{C}$.

An operator $A$ on $\mathbb{C}^{d}$ is called a \textbf{positive operator} if it satisfies the following two conditions:

1) $A$ is a Hermitian operator, i.e., $A=A^{\ast}$, where the transposed conjugate operator $A^{\ast}:=\overline{A^{T}}$ is called the adjoint operator of $A$;

2) $\langle v|A|v\rangle\geq 0$ for any nonzero vector $|v\rangle\in \mathbb{C}^{d}$.

We write $A\geq 0$ if $A$ is a positive operator. A positive operator-valued measure (POVM) is a general measure on a quantum system.
Let us recall its definition.

\begin{definition}\label{definition1.1}
\rm{(\!\!\cite{Davies1976Quantum,Busch1997Operational,Peres2006Quantum,Nielsen2000Quantum})}
Let $X$ be a set. A collection of operators $\{E_{i}\}_{i\in X}$ on $\mathbb{C}^{d}$ is called a \textbf{positive operator-valued measure (POVM)}
if it satisfies
the following two conditions:

(i) $E_{i}\geq 0$ for each $i$;

(ii) $\sum_{i\in X}E_{i}=I_{d}$.
\end{definition}

We call $\rho\in M_{d}(\mathbb{C})$ a $d$-dimensional quantum state (or a qudit) if $\rho\geq 0$ with $\mathrm{Tr}(\rho)=1$.
Let $X$ be a set. Suppose a measurement described by a POVM $\{E_{i}\}_{i\in X}$ is performed in $\rho$.
Then, the probability of obtaining outcome $i\in X$ related to $E_{i}$ is given by the Born rule \cite{Born1955Statistical}:
$p(i)=\mathrm{Tr}(E_i\rho)$. Note that these probabilities $\{p(i)\}_{i\in X}$ must satisfy $\sum_{i\in X}p(i)=1$, which is equivalent to $\sum_{i\in X}E_{i}=I_{d}$ since $\mathrm{Tr}(\rho)=1$.

If the unknown state $\rho$ can be uniquely determined from these probabilities $\{p(i)\}_{i\in X}$, then the POVM $\{E_{i}\}_{i\in X}$ is called an
\textbf{informationally complete POVM (IC-POVM)}. An IC-POVM has the following property.

\begin{proposition}\label{proposition1.2}
\rm{(\!\!\cite[Theorem 2.4]{Ohno2014Necessary})}
\emph{Let $\{E_{i}\}_{i\in X}$ be a POVM on $\mathbb{C}^{d}$.
Then $\{E_{i}\}_{i\in X}$ is an IC-POVM if and only if $\mathrm{span}\{E_{i}\}_{i\in X}=M_{d}(\mathbb{C})$.
}
\end{proposition}

Note that, in fact, sometimes we also refer to the ``if'' part of this proposition as the definition of an IC-POVM
(see, e.g., \cite[Definition 1]{Czerwinski2021Quantum}).
By Proposition \ref{proposition1.2}, we know that the cardinality of an IC-POVM on $\mathbb{C}^{d}$ is at least $d^{2}$.

A particularly appealing and important class of IC-POVMs on $\mathbb{C}^{d}$ are the symmetric informationally complete POVMs (SIC-POVMs),
whose cardinality is exactly $d^{2}$.

\begin{definition}\label{definition1.3}
\rm{(\!\!\cite[pp.2171-2172]{Renes2004Symmetric})}
A POVM $\{E_{1},E_{2},\ldots,E_{d^{2}}\}$ on $\mathbb{C}^{d}$ is called a \textbf{symmetric informationally complete POVM (SIC-POVM)} if it satisfies
the following two coditions:

(i) $E_{i}=\frac{1}{d}|u_{i}\rangle\langle u_{i}|$, where $|u_{i}\rangle$ is a normalized vector in $\mathbb{C}^{d}$ for each $i=1,2,\ldots,d^{2}$;

(ii) $|\langle u_{i}|u_{j}\rangle|^{2}=d^{2}\mathrm{Tr}(E_{i}E_{j})=\frac{1}{d+1}$ for all $1\leq i\neq j\leq d^{2}$.
\end{definition}

According to Definition \ref{definition1.3}, a SIC-POVM on $\mathbb{C}^{d}$ is actually equivalent to
$d^{2}$ pairwise equiangular complex lines, i.e., $d^{2}$ one-dimensional subspaces $\mathbb{C}|u_{1}\rangle,\mathbb{C}|u_{2}\rangle,\ldots,\mathbb{C}|u_{d^{2}}\rangle$ in $\mathbb{C}^{d}$,
which was first studied by Lemmens and Seidal in \cite{Lemmens1973Equiangular}.
Since then, the study of SIC-POVMs and their variants has achieved rapid development
(see, e.g., \cite{Appleby2015Group,Appleby2011The,Delsarte1991Bounds,Hoggar1998Line,Konig1994Norms,Konig1999Cubature,Strohmer2003Grassmannian,
Zauner1999Quantum,Kopp2021SIC-POVMs,Zhu2010SIC,Zhu2018Universally,Tavakoli2020Compounds,Petz2014Conditional,Gour2014Construction,Geng2021What}).
In addition, SIC-POVMs are also closely related to quantum state tomography \cite{Caves2002Unknown,Scott2006Tight},
quantum cryptography \cite{Fuchs2003Squeezing,Fuchs2004On}, design theory \cite{Renes2004Symmetric,Zauner2011Quantum,Klappenecker2005Mutually},
frame theory \cite{Fickus2021Mutually,Fickus2018Tremain,Cahill2018Constructions,Magsino2019Biangular}, etc.

Zauner conjectured in his 1999 PhD thesis \cite{Zauner1999Quantum} that SIC-POVMs on $\mathbb{C}^{d}$ exist in every finite dimension $d\geq 2$.
The analytical constructions of SIC-POVMs on $\mathbb{C}^{d}$ have been given for
$d=2${\color{red}{-}}$24$, $28$, $30$, $31$, $35$, $37$, $39$, $43$, $48$, $124$, $323$
(see \cite{Zauner1999Quantum,Appleby2005Symmetric,Appleby2014Systems,Appleby2012The,Appleby2018Constructing,Grassl2009OnSIC,
Grassl2005Tomography,Grassl2008Computing,Scott2010Symmetric,Grassl2017FibonacciLucas}).
Besides, numerical solutions of SIC-POVMs through computer calculations have been found in all dimensions up to $d=151$,
as well as several other dimensions up to $d=844$
(see \cite{Renes2004Symmetric,Scott2010Symmetric,Grassl2017FibonacciLucas,Fuchs2017The,Scott2017SICs}).

While it has been generally speculated that SIC-POVMs exist in every finite dimension, a rigorous proof of it is still unknown.
Infinite families of SIC-POVMs have not been previously constructed. It is not even clear whether SIC-POVMs on $\mathbb{C}^{d}$
exist for infinitely many dimensions $d$. To realize some possible applications in quantum information processing, Klappenecker \emph{et al.} \cite{Klappenecker2005On} introduced an approximate version of SIC-POVMs by relaxing condition (ii) of Definition \ref{definition1.3} on the inner products slightly in the following definition.

\begin{definition}\label{definition1.4}
\rm{(\!\!\cite[Definition 1]{Klappenecker2005On})}
Suppose that $d$ is a positive integer. Let $\mathcal{A}=\{A_{i}=\frac{1}{d}|u_{i}\rangle\langle u_{i}|:i=1,2,\ldots,d^{2}\}$,
where each $|u_{i}\rangle$ is a normalized vector in $\mathbb{C}^{d}$.
If the following three conditions hold:
\begin{itemize}
\item [(i)] (completeness/POVM condition) $\sum_{i=1}^{d^{2}}A_{i}=I_{d}$;
\vspace{-4pt}

\item [(ii)] (approximate symmetry) $|\langle u_{i}|u_{j}\rangle|^{2}=d^{2}\mathrm{Tr}(A_{i}A_{j})\leq \frac{1+\varepsilon_{d}}{d}$ for all $1\leq i\neq j\leq d^{2}$,
where $\varepsilon_{d}>0$ and $\varepsilon_{d}\rightarrow0$ as $d\rightarrow\infty$;

\item [(iii)] (informational completeness) $A_{1},A_{2},\ldots,A_{d^{2}}$ are linearly independent as elements of $M_{d}(\mathbb{C})$,
\end{itemize}
then the set $\mathcal{A}$ is called an \textbf{approximately symmetric informationally complete POVM
(ASIC-POVM)}. In order to correspond to the term $\varepsilon_{d}$ in condition (ii), we call it an \textbf{$\varepsilon_{d}$-ASIC-POVM} throughout this paper.
\end{definition}

By the conditions (i) and (iii) in Definition \ref{definition1.4}, it is clear that $\varepsilon_{d}$-ASIC-POVMs are a case of IC-POVMs.
Klappenecker \emph{et al.} \cite{Klappenecker2005On} showed that there exist infinite families of $\varepsilon_{d}$-ASIC-POVMs.
To be specific, by using the existence of mutually unbiased bases (MUBs), they gave two constructions of $\varepsilon_{d}$-ASIC-POVMs in dimensions $d=q$ and $d=p$, where $q=p^{k}$ is a prime power. However, it is still difficult to construct $\varepsilon_{d}$-ASIC-POVMs.
After Klappenecker \emph{et al.} \cite{Klappenecker2005On} provided two constructions of $\varepsilon_{d}$-ASIC-POVMs,
only a few families of $\varepsilon_{d}$-ASIC-POVMs were constructed in the literature. These constructions are summarized as follows:
1) Wang \emph{et al.} \cite{Wang2012Constructions} showed a construction of $\varepsilon_{d}$-ASIC-POVMs in dimension $d=q-1$ via Gauss sums;
2) Cao \emph{et al.} \cite{Cao2017Two} obtained two constructions of $\varepsilon_{d}$-ASIC-POVMs in dimensions $d=q$ and $d=q-1$ through MUBs,
perfect nonlinear (PN) functions, permutation polynomials and related character sums;
3) Luo and Cao \cite{Luo2017Two} gave two constructions of $\varepsilon_{d}$-ASIC-POVMs in dimension $d=q$ through Katz sums and Eisenstein sums;
4) Luo and Cao \cite{Luo2018New} provided two constructions of $\varepsilon_{d}$-ASIC-POVMs in dimensions $d=q$ and $d=q+1$ through difference sets of finite abelian groups;
5) Luo \emph{et al.} \cite{Luo2021A} constructed $\varepsilon_{d}$-ASIC-POVMs in dimension $d=q-2$ via Jacobi sums.

Inspired by these excellent works above, in this paper, we construct a class of $\varepsilon_{d}$-ASIC-POVMs in dimension $d=q$ and a class of $\varepsilon_{d}$-ASIC-POVMs in dimension $d=q+1$, respectively, where $q$ is a prime power.
In the first construction, we prove that all $2$-to-$1$ PN functions can be used for constructing $\varepsilon_{q}$-ASIC-POVMs (see Theorem \ref{theorem2.10}).
We show that there exist $2$-to-$1$ PN functions that do not satisfy the condition in \cite[Theorem III.3]{Cao2017Two}; see Remark \ref{remark2.11}.
Thus, the class of $\varepsilon_{q}$-ASIC-POVMs constructed in this paper is more general than the class constructed in \cite{Cao2017Two}.
We also show that the set of vectors corresponding to the $\varepsilon_{q}$-ASIC-POVM forms a biangular frame (see Theorem \ref{theorem2.15}).
The construction of $\varepsilon_{q+1}$-ASIC-POVMs is based on a multiplicative character sum estimate called the Li bound (see Theorem \ref{theorem3.5}).
We show that the set of vectors corresponding to the $\varepsilon_{q+1}$-ASIC-POVM forms an asymptotically optimal codebook (see subsection \ref{subsection 3.3}).
In Theorem \ref{theorem2.10} (resp. Theorem \ref{theorem3.5}), we also characterize ``how close'' the $\varepsilon_{q}$-ASIC-POVMs (resp. $\varepsilon_{q+1}$-ASIC-POVMs) are from being SIC-POVMs of dimension $q$ (resp. dimension $q+1$). Finally, we explain the significance of constructing $\varepsilon_{d}$-ASIC-POVMs (see Section \ref{section4}).

The rest of this paper is organized as follows.
In Section \ref{section2}, we present the construction of $\varepsilon_{q}$-ASIC-POVMs by using general $2$-to-$1$ PN functions,
and we also show that the corresponding set of vectors forms a biangular frame.
In Section \ref{section3}, we provide the construction of $\varepsilon_{q+1}$-ASIC-POVMs based on the Li bound,
and we also show that the corresponding set of vectors forms an asymptotically optimal codebook.
In Section \ref{section4}, we give a conclusion of this paper.

\section{On $\varepsilon_{q}$-ASIC-POVMs and corresponding set of vectors by using $2$-to-$1$ PN functions}\label{section2}

In this section, we will present the construction of $\varepsilon_{q}$-ASIC-POVMs by using general $2$-to-$1$ PN functions.
From this construction, we are able to obtain more candidate $2$-to-$1$ PN functions for constructing $\varepsilon_{q}$-ASIC-POVMs.
Besides, we will also show that the set of vectors corresponding to the $\varepsilon_{q}$-ASIC-POVM forms a biangular frame.

\subsection{Some notation and a useful lemma}\label{subsection2.1}

First, let us recall some basic notation and results on characters and character groups of a finite field.
Let $q=p^{k}$ be a prime power. Set $\zeta_{p}=e^{\frac{2\pi i}{p}}$. Then, for each $a\in\mathbb{F}_{q}$, the mapping
\begin{equation}\label{eq2.1}
\chi_{a}(b)=\zeta_{p}^{\mathrm{tr}_{\mathbb{F}_{q}/\mathbb{F}_{p}}(ab)},\ b\in\mathbb{F}_{q}
\end{equation}
defines an additive character of $\mathbb{F}_{q}$, where $\mathrm{tr}_{\mathbb{F}_{q}/\mathbb{F}_{p}}(x):=\sum_{i=0}^{k-1}x^{p^{i}}$
represents the trace mapping from $\mathbb{F}_{q}$ to $\mathbb{F}_{p}$.
Clearly, $\chi_{0}(b)=1$ for each $b\in\mathbb{F}_{q}$. We call $\chi_{0}$ the trivial character of $\mathbb{F}_{q}$.
For each character $\chi_{a}$ of $\mathbb{F}_{q}$, we define its conjugate $\overline{\chi_{a}}$ as
$\overline{\chi_{a}}(b)=\overline{\chi_{a}(b)}$ for all $b\in\mathbb{F}_{q}$.
Then, $\overline{\chi_{a}}$ is a character of $\mathbb{F}_{q}$.
The set $\{\chi_{a}|a\in\mathbb{F}_{q}\}$ forms an abelian group, which is called the additive character group of $\mathbb{F}_{q}$
and denoted by $\widehat{\mathbb{F}_{q}}$.
One can verify that the map $(a\mapsto \chi_{a})$ defines an isomorphism from $\mathbb{F}_{q}$ to $\widehat{\mathbb{F}_{q}}$.

Based on these facts, one has the following properties (see, e.g., \cite{Lidl1997Finite}):
\begin{equation}\label{eq2.2}
\sum_{b\in\mathbb{F}_{q}}\chi_{a}(b)=\begin{cases}
q,&\mbox{if}\ a=0,\\
0,&\mbox{if}\ a\neq0,
\end{cases}
\end{equation}
and
\begin{equation}\label{eq2.3}
\sum_{a\in\mathbb{F}_{q}}\chi_{a}(b)=\begin{cases}
q,&\mbox{if}\ b=0,\\
0,&\mbox{if}\ b\neq0.
\end{cases}
\end{equation}

The following lemma will be used in the rest of this section.

\begin{lemma}\label{lemma2.1}
Let $q=p^{k}$ be a prime power. Suppose that $y_{a}\in\mathbb{C}$. Then,
$\sum_{a\in\mathbb{F}_{q}}y_{a}\chi_{a}(b)=0$ for all $b\in\mathbb{F}_{q}^{\ast}$ if and only if
there exists $m\in\mathbb{C}$ such that $y_{a}\equiv m$ for all $a\in\mathbb{F}_{q}$.
\end{lemma}

\begin{proof}
``$\Longleftarrow$'': This is due to the fact that $\sum_{a\in\mathbb{F}_{q}}\chi_{a}(b)=0$
for all $b\in\mathbb{F}_{q}^{\ast}$, as shown in Eq. (\ref{eq2.3}).

``$\Longrightarrow$'': Define a Fourier transformation $F(b)=\sum_{a\in\mathbb{F}_{q}}y_{a}\chi_{a}(b)$, where $b\in\mathbb{F}_{q}$.
Then $F(0)=\sum_{a\in\mathbb{F}_{q}}y_{a}$ and $F(b)=0$ for all $b\in\mathbb{F}_{q}^{\ast}$. Thus, for any $a\in\mathbb{F}_{q}$, we have
\begin{align*}
y_{a}
&=\frac{1}{q}\sum_{b\in\mathbb{F}_{q}}F(b)\overline{\chi_{a}}(b)\\
&=\frac{1}{q}\left(F(0)\overline{\chi_{a}}(0)+\sum_{b\in\mathbb{F}_{q}^{\ast}}F(b)\overline{\chi_{a}}(b)\right)\\
&=\frac{1}{q}\sum_{a\in\mathbb{F}_{q}}y_{a}.
\end{align*}

Therefore, the proof is completed. $\hfill\square$
\end{proof}

\subsection{Construction of $\varepsilon_{q}$-ASIC-POVMs}\label{subsection2.2}

In \cite[Theorem III.3]{Cao2017Two}, the authors constructed a family of $\varepsilon_{q}$-ASIC-POVMs by
using a special kind of $2$-to-$1$ PN functions $f(x)$ that satisfy $f(0)=0$ and $f(x)=f(y)$ iff $x=-y$ for all $x,y\in\mathbb{F}_{q}$.
In this subsection, we will show that all $2$-to-$1$ PN functions can be utilized for constructing
$\varepsilon_{q}$-ASIC-POVMs.

First, let us recall the concepts of $2$-to-$1$ mappings and perfect nonlinear (PN) functions.

\begin{definition}\label{definition2.3}
\rm{(\!\!\cite[Definition 1]{Mesnager2019On})}
Let $f$ be a mapping from a finite set $A$ to a finite set $B$. We call $f$ a \textbf{$2$-to-$1$ mapping}
if it satisfies one of the following properties:
\vspace{-3pt}
\begin{itemize}
\item [1)] $\sharp A$ is even, and for any element of $B$, it has either $2$ or $0$ preimages of $f$;

\vspace{-3pt}

\item [2)] $\sharp A$ is odd, and for all but one element of $B$, it has either $2$ or $0$
preimages of $f$, and the exceptional element has exactly one preimage.
\end{itemize}
\end{definition}

\begin{remark}\label{remark2.4}
One can calculate the cardinality of all $2$-to-$1$ mappings over $\mathbb{F}_{q}$; see Proposition \ref{propositionE} of Appendix. For more properties on $2$-to-$1$ mappings and related applications, we refer the reader to the latest references \cite{Li2021Further,Mesnager2023Further}.
\end{remark}

\begin{remark}\label{remark2.5}
Let $q$ be an odd prime power. Then for any $2$-to-$1$ mapping $f(x)$ over $\mathbb{F}_{q}$,
we can write $\mathbb{F}_{q}$ as $\mathbb{F}_{q}=\{a_{1},a_{2},\ldots,a_{q}\}$ such that
\begin{equation}\label{eq2.8}
f(a_{i})=f(a_{q+2-i}),\ i=2,3,\ldots,\frac{q+1}{2};\ f(a_{1})\notin\left\{f(a_{i})\Big|i=2,3,\ldots,\frac{q+1}{2}\right\}.
\end{equation}
Here, Eq. (\ref{eq2.8}) is called a \emph{$f$-permutation} of $\mathbb{F}_{q}$.
\end{remark}

\begin{definition}\label{definition2.6}
\rm{(\!\!\cite[Definition 2.1]{Davis2008G})}
Let $G$ and $H$ be abelian groups. Let $f: G\rightarrow H$ be a function from $G$ and $H$.
Then the function $f$ is called \textbf{perfect nonlinear (PN)} if
$|\{x\in G|f(x+a)-f(x)=b\}|=\frac{|G|}{|H|}$ for every $(a,b)\in G\setminus\{0\}\times H$.
\end{definition}

\begin{remark}\label{remark2.7}
For any PN function $f(x)$ from $\mathbb{F}_{q}$ to $\mathbb{F}_{q}$, it is easy to see that $\mathbb{F}_{q}=\{f(x+a)-f(x)|x\in\mathbb{F}_{q}\}$
for every $a\in \mathbb{F}_{q}^{\ast}$ (or equivalently, $f(x)$ is a PN function if and only if $f(x+a)-f(x)$ is bijective for every $a\in \mathbb{F}_{q}^{\ast}$).
For the surveys on PN functions (and, more generally, highly nonlinear functions), we refer the reader to \cite{Blondeau2015Perfect,Carlet2004Highly}.
\end{remark}

The following result related to PN functions is a corollary of \cite[Lemma 5.5]{Hall2011Mutually} (see also \cite[Lemma III.2]{Cao2017Two}),
which will be used later in the construction of $\varepsilon_{q}$-ASIC-POVMs.

\begin{lemma}\label{lemma2.8}
Let $f(x)\in\mathbb{F}_{q}[x]$ be a PN function, and let $\chi$ be an arbitrary nontrivial additive character of $\mathbb{F}_{q}$. Then
$$\Big|\sum_{x\in \mathbb{F}_{q}}\chi (af(x)+bx)\Big|=\sqrt{q}$$
for all $a\in\mathbb{F}_{q}^{\ast}$ and $b\in\mathbb{F}_{q}$.
\end{lemma}

Let us give a useful property on PN functions in the following lemma.

\begin{lemma}\label{lemma2.9}
Let $\mathbb{F}_{q}=\{a_{1},a_{2},\ldots,a_{q}\}$ and let $f(x)\in\mathbb{F}_{q}[x]$ be a PN function.
For an arbitrary nontrivial additive character $\chi$ of $\mathbb{F}_{q}$, if
\begin{equation}\label{eq2.9}
\sum_{a\in \mathbb{F}_{q}}\sum_{b\in\mathbb{F}_{q}^{\ast}}k_{a,b}\chi\Big(a\big(f(a_{i})-f(a_{j})\big)+b(a_{i}-a_{j})\Big)=0
\end{equation}
for all $1\leq i\neq j\leq q$, then $k_{a,b}\equiv0$ for all $a\in\mathbb{F}_{q}$ and $b\in\mathbb{F}_{q}^{\ast}$.
\end{lemma}

\begin{proof}
We may write $\chi=\chi_{r}$ for some $0\neq r\in\mathbb{F}_{q}$, where $\chi_{r}$ is defined as in Eq. (\ref{eq2.1}).
Put $k_{a,0}=0$ for all $a\in\mathbb{F}_{q}$, then Eq. (\ref{eq2.9}) is equivalent to
\begin{equation}\label{eq2.10}
\sum_{a\in\mathbb{F}_{q}}\chi_{r}\Big(a\big(f(a_{i})-f(a_{j})\big)\Big)\sum_{b\in\mathbb{F}_{q}}k_{a,b}\chi_{r}\big(b(a_{i}-a_{j})\big)=0
\end{equation}
for all distinct elements $a_{i}$ and $a_{j}$ in $\mathbb{F}_{q}$.
For any $c\in\mathbb{F}_{q}^{\ast}$ and $d\in\mathbb{F}_{q}$, by the definition of PN functions over $\mathbb{F}_{q}$,
there exists $m\in\mathbb{F}_{q}$ such that $f(m+c)-f(m)=d$. That is, for any $c\in\mathbb{F}_{q}^{\ast}$ and $d\in\mathbb{F}_{q}$,
there exist $a_{s}:=m+c$ and $a_{t}:=m$ such that
\begin{align*}
0
&=\sum_{a\in\mathbb{F}_{q}}\chi_{r}\Big(a\big(f(a_{s})-f(a_{t})\big)\Big)\sum_{b\in\mathbb{F}_{q}}k_{a,b}\chi_{r}\big(b(a_{s}-a_{t})\big)\\
&=\sum_{a\in\mathbb{F}_{q}}\chi_{r}(ad)\sum_{b\in\mathbb{F}_{q}}k_{a,b}\chi_{r}(bc).
\end{align*}

Set $h_{a,c}=\sum_{b\in\mathbb{F}_{q}^{\ast}}k_{a,b}\chi_{r}(bc)$. Then, for all $c\in\mathbb{F}_{q}^{\ast}$ and $d\in\mathbb{F}_{q}$,
\begin{equation}\label{eq2.11}
\sum_{a\in\mathbb{F}_{q}}h_{a,c}\chi_{r}(ad)=0\Longrightarrow\sum_{a\in\mathbb{F}_{q}}h_{\frac{a}{r},c}\chi_{a}(d)=0.
\end{equation}
For any fixed $c_{0}\in\mathbb{F}_{q}^{\ast}$, by Lemma \ref{lemma2.1} there exists $m\in\mathbb{C}$ such that $h_{\frac{a}{r},c_{0}}\equiv m$ for all $a\in\mathbb{F}_{q}$. Further, by taking $c=c_{0}$ and $d=0$ in Eq. (\ref{eq2.11}) we obtain
$qm=0$, then $m=0$, i.e., $h_{\frac{a}{r},c_{0}}\equiv0$ for all $a\in\mathbb{F}_{q}$.
This implies that $h_{a,c_{0}}\equiv0$ and thus $h_{a,c}\equiv0$ for all $a\in\mathbb{F}_{q}$ and $c\in\mathbb{F}_{q}^{\ast}$. Namely,
$\sum_{b\in\mathbb{F}_{q}^{\ast}}k_{a,b}\chi_{r}(bc)=\sum_{b\in\mathbb{F}_{q}}k_{a,b}\chi_{r}(bc)=0$, which can be rewritten as
\begin{equation*}
\sum_{b\in\mathbb{F}_{q}}k_{a,b}\chi_{b}(rc)=0
\end{equation*}
for all $a\in\mathbb{F}_{q}$ and $c\in\mathbb{F}_{q}^{\ast}$.
For any fixed $a_{0}\in\mathbb{F}_{q}$, by using Lemma \ref{lemma2.1} again, there exists $\eta\in\mathbb{C}$ such that $k_{a_{0},b}\equiv\eta$ for all $b\in\mathbb{F}_{q}$. Since $k_{a_{0},0}=0$, we obtain $k_{a_{0},b}\equiv0$ for all $b\in\mathbb{F}_{q}$.
Therefore, $k_{a,b}\equiv0$ for all $a\in\mathbb{F}_{q}$ and $b\in\mathbb{F}_{q}$, which completes the proof.  $\hfill\square$
\end{proof}

\vspace{6pt}

By using general $2$-to-$1$ PN functions (i.e., function that are both $2$-to-$1$ mappings and PN functions),
let us provide the corresponding construction of $\varepsilon_{q}$-ASIC-POVMs in the following theorem.

\begin{theorem}\label{theorem2.10}
Let $q$ be an odd prime power. Let $f(x)\in\mathbb{F}_{q}[x]$ be any $2$-to-$1$
PN function, and let Eq. (\ref{eq2.8}) be a $f$-permutation of $\mathbb{F}_{q}=\{a_{1},a_{2},\ldots,a_{q}\}$.
Suppose $\chi$ is an arbitrary nontrivial additive character of $\mathbb{F}_{q}$.
For each $a\in \mathbb{F}_{q}$ and $b\in\mathbb{F}_{q}^{\ast}$, define $|v_{a,b}\rangle=\frac{1}{\sqrt{q}}\big(\chi(af(a_{i})+ba_{i})\big)^{T}_{i=1,\ldots,q}\in \mathbb{C}^{q}$.
Define
\begin{equation}\label{eq2.12}
\mathcal{A}=\{|v_{a,b}\rangle:a\in\mathbb{F}_{q},\ b\in\mathbb{F}_{q}^{\ast}\}\cup \{|e_{1}\rangle\cup\ldots\cup|e_{q}\rangle\},
\end{equation}
where $|e_{i}\rangle\in\mathbb{C}^{q}$ is the unit vector whose $i$-th component is $1$, and $0$ in other components.
Denote $E=\sum_{i=1}^{q^{2}}E_{i}$ with $E_{i}:=\frac{1}{q}|v_{i}\rangle\langle v_{i}|$ for $|v_{i}\rangle\in\mathcal{A}$.
Then $\mathcal{M}=\{M_{i}:=E^{-\frac{1}{2}}E_{i}E^{-\frac{1}{2}}|i=1,2,\ldots,q^{2}\}$ forms an $\varepsilon_{q}$-ASIC-POVM, where

1) For $M_{i}=\frac{1}{q}E^{-\frac{1}{2}}|v_{a,b}\rangle\langle v_{a,b}|E^{-\frac{1}{2}}$ and
$M_{j}=\frac{1}{q}E^{-\frac{1}{2}}|v_{c,d}\rangle\langle v_{c,d}|E^{-\frac{1}{2}}$ with $a\neq c$,
\begin{equation*}
q^{2}\mathrm{Tr}(M_{i}M_{j})\leq \Bigg(\frac{q^{2}\sqrt{q}+q^{2}-q+1}{(q^{2}-1)q}\Bigg)^{2}=\frac{1+\varepsilon_{q}}{q},
\end{equation*}
where the infinitesimal $\varepsilon_{q}$ is given by
\begin{equation*}
\varepsilon_{q}=\frac{(q^{2}-q+\sqrt{q}+1)(2q^{2}\sqrt{q}+q^{2}-q-\sqrt{q}+1)}{q(q^{2}-1)^{2}}.
\end{equation*}
Besides, we have
\begin{equation*}
\left|\frac{1}{q+1}-\left(\frac{q^{2}\sqrt{q}+q^{2}-q+1}{(q^{2}-1)q}\right)^{2}\right|=O(q^{-\frac{3}{2}}).
\end{equation*}

\vspace{6pt}

2) For $M_{i}=\frac{1}{q}E^{-\frac{1}{2}}|v_{a,b}\rangle\langle v_{a,b}|E^{-\frac{1}{2}}$ and
$M_{j}=\frac{1}{q}E^{-\frac{1}{2}}|v_{c,d}\rangle\langle v_{c,d}|E^{-\frac{1}{2}}$ with $a=c$ and $b\neq d$,
\begin{equation*}
q^{2}\mathrm{Tr}(M_{i}M_{j})\leq \Bigg(\frac{q^{2}-q+1}{(q^{2}-1)q}\Bigg)^{2}\leq\frac{1+\varepsilon_{q}}{q},
\end{equation*}
where the infinitesimal $\varepsilon_{q}$ is given by
\begin{equation*}
\varepsilon_{q}=\frac{(q^{2}-q+1)(2q^{2}\sqrt{q}+q^{2}-q-2\sqrt{q}+1)}{q(q^{2}-1)^{2}}.
\end{equation*}
Besides, we have
\begin{equation*}
\left|\frac{1}{q+1}-\left(\frac{q^{2}-q+1}{(q^{2}-1)q}\right)^{2}\right|=O(q^{-1}).
\end{equation*}

\vspace{6pt}

3) For $M_{i}=\frac{1}{q}E^{-\frac{1}{2}}|v_{a,b}\rangle\langle v_{a,b}|E^{-\frac{1}{2}}$ and
$M_{j}=\frac{1}{q}E^{-\frac{1}{2}}|e_{j}\rangle\langle e_{j}|E^{-\frac{1}{2}}$,
\begin{equation*}
q^{2}\mathrm{Tr}(M_{i}M_{j})\leq \left(\frac{q^{2}+q+1}{(q^{2}-1)\sqrt{q}}\right)^{2}=\frac{1+\varepsilon_{q}}{q},
\end{equation*}
where the infinitesimal $\varepsilon_{q}$ is given by
\begin{equation*}
\varepsilon_{q}=\frac{q(q+2)(2q+1)}{(q^{2}-1)^{2}}.
\end{equation*}
Besides, we have
\begin{equation*}
\left|\frac{1}{q+1}-\left(\frac{q^{2}+q+1}{(q^{2}-1)\sqrt{q}}\right)^{2}\right|=O(q^{-2}).
\end{equation*}

\vspace{6pt}

4) For $M_{i}=\frac{1}{q}E^{-\frac{1}{2}}|e_{i}\rangle\langle e_{i}|E^{-\frac{1}{2}}$ and
$M_{j}=\frac{1}{q}E^{-\frac{1}{2}}|e_{j}\rangle\langle e_{j}|E^{-\frac{1}{2}}$ with $i\neq j$,
\begin{equation*}
q^{2}\mathrm{Tr}(M_{i}M_{j})\leq \Big(\frac{q}{q^{2}-1}\Big)^{2}\leq \frac{1+\varepsilon_{q}}{q},
\end{equation*}
where the infinitesimal $\varepsilon_{q}$ is given by
\begin{equation*}
\varepsilon_{q}=\frac{q\sqrt{q}(2q^{2}+q\sqrt{q}-2)}{(q^{2}-1)^{2}}.
\end{equation*}
Besides, we have
\begin{equation*}
\left|\frac{1}{q+1}-\left(\frac{q}{q^{2}-1}\right)^{2}\right|=O(q^{-1}).
\end{equation*}
\end{theorem}

\begin{proof}
Since Eq. (\ref{eq2.8}) is a $f$-permutation of $\mathbb{F}_{q}=\{a_{1},a_{2},\ldots,a_{q}\}$, we have
$f(a_{i})=f(a_{q+2-i})$ for $i=2,3,\ldots,\frac{q+1}{2}$ and $f(a_{1})\notin\{f(a_{i})|i=2,3,\ldots,\frac{q+1}{2}\}$.

First, let us prove that $\mathcal{M}$ satisfies the completeness/POVM condition of an $\varepsilon_{q}$-ASIC-POVM. We have
\begin{equation*}
E=\sum_{i=1}^{q^{2}}E_{i}=\frac{1}{q}\sum_{a\in \mathbb{F}_{q}}\sum_{b\in \mathbb{F}_{q}^{\ast}}|v_{a,b}\rangle\langle v_{a,b}|+\frac{1}{q}I_{q}.
\end{equation*}
For any $1\leq i,j\leq q$, the $(i,j)$-th entry of $E$ is
\begin{equation*}
E(i,j)=\frac{1}{q^{2}}\sum_{a\in \mathbb{F}_{q}}\chi\Big(a\big(f(a_{i})-f(a_{j})\big)\Big)\sum_{b\in \mathbb{F}_{q}^{\ast}}\chi\big(b(a_{i}-a_{j})\big)+\frac{\delta_{i,j}}{q}.
\end{equation*}
Hence, we know that

(a) For $i=1,2,\ldots,q$, we have $E(i,i)=\frac{1}{q^{2}}\cdot q\cdot (q-1)+\frac{1}{q}=1$.

(b) For $i=2,3,\ldots,\frac{q+1}{2}$, we have
\begin{equation*}
E(i,q+2-i)=\frac{1}{q^{2}}\sum_{a\in \mathbb{F}_{q}}\chi\Big(a\big(f(a_{i})-f(a_{q+2-i})\big)\Big)\sum_{b\in \mathbb{F}_{q}^{\ast}}\chi\big(b(a_{i}-a_{q+2-i})\big)=\frac{1}{q^{2}}\cdot q\cdot (-1)=-\frac{1}{q}.
\end{equation*}

(c) For $i\neq j$ and $i+j\neq q+2$, we have $E(i,j)=0$.

Therefore, we obtain that
\begin{align*}
E=\left[
\begin{array}{cc}
1&\mathbf{0}_{1\times(q-1)}\\
\mathbf{0}_{(q-1)\times1}&F_{q-1}\\
\end{array}\right],
\end{align*}
where
\begin{align}\label{eq2.13}
F_{q-1}=\left[
\begin{array}{cc}
I_{\frac{q-1}{2}}&-\frac{1}{q}J_{\frac{q-1}{2}}\\
-\frac{1}{q}J_{\frac{q-1}{2}}&I_{\frac{q-1}{2}}\\
\end{array}\right]
\end{align}
in which $I_{\frac{q-1}{2}}$ is an $\frac{q-1}{2}\times\frac{q-1}{2}$ identity matrix, and
$J_{\frac{q-1}{2}}:=(\delta_{i,\frac{q+1}{2}-i})_{i=1}^{\frac{q-1}{2}}$ is an $\frac{q-1}{2}\times\frac{q-1}{2}$ flip matrix.
By $\mathrm{det}(F_{q-1})=(1-\frac{1}{q^{2}})\mathrm{det}(F_{q-3})$, we know that
$\mathrm{det}(F_{q-1})=(1-\frac{1}{q^{2}})^{\frac{q-1}{2}}$.
It is straightforward to verify that all leading principal minors of $F_{q-1}$ are positive. Then, $F_{q-1}$ is positive definite, and so is $E$.
Hence, $E^{-1}$ exists and it is positive definite as well, which implies that we can obtain the uniquely determined positive definite matrix $E^{-\frac{1}{2}}$, and thus
\begin{equation*}
\sum_{i=1}^{q^{2}}M_{i}=E^{-\frac{1}{2}}\Big(\sum_{i=1}^{q^{2}}E_{i}\Big)E^{-\frac{1}{2}}=I_{q}.
\end{equation*}
Therefore, $\mathcal{M}$ satisfies the completeness/POVM condition of an $\varepsilon_{q}$-ASIC-POVM.

Next, let us prove that $\mathcal{M}$ satisfies the approximate symmetry condition of an $\varepsilon_{q}$-ASIC-POVM.
It follows from Eq. (\ref{eq2.13}) that
\begin{equation*}
F_{q-1}^{-1}=\left[
\begin{array}{cc}
\frac{q^{2}}{q^{2}-1}I_{\frac{q-1}{2}}&\frac{q}{q^{2}-1}J_{\frac{q-1}{2}}\\
\frac{q}{q^{2}-1}J_{\frac{q-1}{2}}&\frac{q^{2}}{q^{2}-1}I_{\frac{q-1}{2}}\\
\end{array}\right].
\end{equation*}
Then, we have
\begin{align}\label{eq2.14}
E^{-1}&= \left[
\begin{array}{cc}
1&\mathbf{0}_{1\times(q-1)}\\
\mathbf{0}_{(q-1)\times1}&F_{q-1}^{-1}\\
\end{array}\right]\nonumber\\
&=\frac{q^{2}}{q^{2}-1}I_{q}-\frac{1}{q^{2}-1}R_{1,1}+\frac{q}{q^{2}-1}\sum_{i=2}^{q}R_{i,q+2-i},
\end{align}
where $R_{i,j}$ denotes the $q\times q$ matrix whose $(i,j)$-th entry is $1$, and $0$ elsewhere.

{\bfseries Case 1:} For $M_{i}=\frac{1}{q}E^{-\frac{1}{2}}|v_{a,b}\rangle\langle v_{a,b}|E^{-\frac{1}{2}}$ and
$M_{j}=\frac{1}{q}E^{-\frac{1}{2}}|v_{c,d}\rangle\langle v_{c,d}|E^{-\frac{1}{2}}$, where $(a,b)\neq(c,d)$, we have
\begin{equation*}
\begin{split}
q^{2}\mathrm{Tr}(M_{i}M_{j})
&=|\langle v_{a,b}|E^{-1}|v_{c,d}\rangle|^{2}\\\nonumber
&=\Big|\frac{q^{2}}{q^{2}-1}\langle v_{a,b}|v_{c,d}\rangle-\frac{1}{q^{2}-1}\langle v_{a,b}|R_{1,1}|v_{c,d}\rangle+\frac{q}{q^{2}-1}\sum_{i=2}^{q}\langle v_{a,b}|R_{i,q+2-i}|v_{c,d}\rangle\Big|^{2}.
\end{split}
\end{equation*}
One can calculate that
\begin{equation*}
\langle v_{a,b}|v_{c,d}\rangle=\frac{1}{q}\sum_{i=1}^{q}\chi\big((c-a)f(a_{i})+(d-b)a_{i}\big),
\end{equation*}
\begin{equation*}
\big|\langle v_{a,b}|R_{1,1}|v_{c,d}\rangle\big|=\frac{1}{q}\Big|\chi\big((c-a)f(a_{1})+(d-b)a_{1}\big)\Big|\leq \frac{1}{q},
\end{equation*}
\begin{equation*}
\Big|\sum_{i=2}^{q}\langle v_{a,b}|R_{i,q+2-i}|v_{c,d}\rangle\Big|=\frac{1}{q}\Big|\sum_{i=2}^{q}\chi\big(cf(a_{q+2-i})+da_{q+2-i}-af(a_{i})-ba_{i}\big)\Big|
\leq \frac{q-1}{q}.
\end{equation*}

{\bfseries Case 1.1:} When $a\neq c$, it follows from Lemma \ref{lemma2.8} that $\big|\langle v_{a,b}|v_{c,d}\rangle\big|=\frac{1}{q}\cdot\sqrt{q}=\frac{1}{\sqrt{q}}$.
In this case,
\begin{align*}
q^{2}\mathrm{Tr}(M_{i}M_{j})
&\leq \Big(\frac{q^{2}}{q^{2}-1}\cdot \frac{1}{\sqrt{q}}+\frac{1}{q^{2}-1}\cdot \frac{1}{q}+\frac{q}{q^{2}-1}\cdot \frac{q-1}{q}\Big)^{2}\\
&=\Bigg(\frac{q^{2}\sqrt{q}+q^{2}-q+1}{(q^{2}-1)q}\Bigg)^{2}\\
&=\Bigg(\frac{\frac{q^{2}}{q^{2}-1}+\frac{1}{\sqrt{q}(q^{2}-1)}+\frac{\sqrt{q}}{q+1}}{\sqrt{q}}\Bigg)^{2}
=\left(\frac{1+\tilde{\varepsilon}_{q}}{\sqrt{q}}\right)^{2}=\frac{1+\varepsilon_{q}}{q},
\end{align*}
where the two infinitesimals $\tilde{\varepsilon}_{q}$ and $\varepsilon_{q}$ are given by
\begin{equation*}
\tilde{\varepsilon}_{q}=\frac{q^{2}-q+\sqrt{q}+1}{\sqrt{q}(q^{2}-1)},
\end{equation*}
\begin{equation*}
\varepsilon_{q}=\tilde{\varepsilon}_{q}^{2}+2\tilde{\varepsilon}_{q}
=\frac{(q^{2}-q+\sqrt{q}+1)(2q^{2}\sqrt{q}+q^{2}-q-\sqrt{q}+1)}{q(q^{2}-1)^{2}}.
\end{equation*}
Besides, we have
\begin{equation*}
\left|\frac{1}{q+1}-\left(\frac{q^{2}\sqrt{q}+q^{2}-q+1}{(q^{2}-1)q}\right)^{2}\right|=O(q^{-\frac{3}{2}}).
\end{equation*}

{\bfseries Case 1.2:} When $a=c$ and $b\neq d$, we have $\langle v_{a,b}|v_{c,d}\rangle=\frac{1}{q}\sum_{i=1}^{q}\chi\big((d-b)a_{i}\big)=0$.
Similar to {\bfseries Case 1.1}, we immediately obtain
\begin{align*}
q^{2}\mathrm{Tr}(M_{i}M_{j})
&\leq \Big(\frac{1}{q^{2}-1}\cdot \frac{1}{q}+\frac{q}{q^{2}-1}\cdot \frac{q-1}{q}\Big)^{2}\\
&=\Bigg(\frac{\frac{1}{\sqrt{q}(q^{2}-1)}+\frac{\sqrt{q}}{q+1}}{\sqrt{q}}\Bigg)^{2}\\
&\leq\Bigg(\frac{1+\frac{1}{\sqrt{q}(q^{2}-1)}+\frac{\sqrt{q}}{q+1}}{\sqrt{q}}\Bigg)^{2}
=\left(\frac{1+\tilde{\varepsilon}_{q}}{\sqrt{q}}\right)^{2}=\frac{1+\varepsilon_{q}}{q},
\end{align*}
where the two infinitesimals $\tilde{\varepsilon}_{q}$ and $\varepsilon_{q}$ are given by
\begin{equation*}
\tilde{\varepsilon}_{q}=\frac{q^{2}-q+1}{\sqrt{q}(q^{2}-1)},
\end{equation*}
\begin{equation*}
\varepsilon_{q}=\tilde{\varepsilon}_{q}^{2}+2\tilde{\varepsilon}_{q}
=\frac{(q^{2}-q+1)(2q^{2}\sqrt{q}+q^{2}-q-2\sqrt{q}+1)}{q(q^{2}-1)^{2}}.
\end{equation*}
Besides, we have
\begin{equation*}
\left|\frac{1}{q+1}-\left(\frac{q^{2}-q+1}{(q^{2}-1)q}\right)^{2}\right|=O(q^{-1}).
\end{equation*}

{\bfseries Case 2:} For $M_{i}=\frac{1}{q}E^{-\frac{1}{2}}|v_{a,b}\rangle\langle v_{a,b}|E^{-\frac{1}{2}}$ and
$M_{j}=\frac{1}{q}E^{-\frac{1}{2}}|e_{j}\rangle\langle e_{j}|E^{-\frac{1}{2}}$, we obtain that
\begin{align*}
q^{2}\mathrm{Tr}(M_{i}M_{j})
&=|\langle e_{j}|E^{-1}|v_{a,b}\rangle|^{2}\\
&=\Big|\frac{q^{2}}{q^{2}-1}\langle e_{j}|v_{a,b}\rangle-\frac{1}{q^{2}-1}\langle e_{j}|R_{1,1}|v_{a,b}\rangle+\frac{q}{q^{2}-1}\langle e_{j}|\sum_{i=2}^{q}R_{i,q+2-i}|v_{a,b}\rangle\Big|^{2}\\
&\leq \Big(\frac{q^{2}}{q^{2}-1}\cdot \frac{1}{\sqrt{q}}+\frac{1}{q^{2}-1}\cdot \frac{1}{\sqrt{q}}+\frac{q}{q^{2}-1}\cdot \frac{1}{\sqrt{q}}\Big)^{2}\\
&=\left(\frac{q^{2}+q+1}{(q^{2}-1)\sqrt{q}}\right)^{2}=\Bigg(\frac{\frac{q^{2}+q+1}{q^{2}-1}}{\sqrt{q}}\Bigg)^{2}
=\left(\frac{1+\tilde{\varepsilon}_{q}}{\sqrt{q}}\right)^{2}=\frac{1+\varepsilon_{q}}{q},
\end{align*}
where the two infinitesimals $\tilde{\varepsilon}_{q}$ and $\varepsilon_{q}$ are given by
\begin{equation*}
\tilde{\varepsilon}_{q}=\frac{q+2}{q^{2}-1},
\end{equation*}
\begin{equation*}
\varepsilon_{q}=\tilde{\varepsilon}_{q}^{2}+2\tilde{\varepsilon}_{q}=\frac{q(q+2)(2q+1)}{(q^{2}-1)^{2}}.
\end{equation*}
Besides, we have
\begin{equation*}
\left|\frac{1}{q+1}-\left(\frac{q^{2}+q+1}{(q^{2}-1)\sqrt{q}}\right)^{2}\right|=O(q^{-2}).
\end{equation*}

{\bfseries Case 3:} For $M_{i}=\frac{1}{q}E^{-\frac{1}{2}}|e_{i}\rangle\langle e_{i}|E^{-\frac{1}{2}}$ and
$M_{j}=\frac{1}{q}E^{-\frac{1}{2}}|e_{j}\rangle\langle e_{j}|E^{-\frac{1}{2}}$, where $i\neq j$, we have
\begin{align*}
q^{2}\mathrm{Tr}(M_{i}M_{j})
&=|\langle e_{j}|E^{-1}|e_{i}\rangle|^{2}\\
&=\Big|\frac{q^{2}}{q^{2}-1}\langle e_{j}|e_{i}\rangle-\frac{1}{q^{2}-1}\langle e_{j}|R_{1,1}|e_{i}\rangle+\frac{q}{q^{2}-1}\langle e_{j}|\sum_{i=2}^{q}R_{i,q+2-i}|e_{i}\rangle\Big|^{2}\\
&\leq \Big(\frac{q^{2}}{q^{2}-1}\cdot 0+\frac{1}{q^{2}-1}\cdot 0+\frac{q}{q^{2}-1}\Big)^{2}\\
&=\left(\frac{q}{q^{2}-1}\right)^{2}=\Bigg(\frac{\frac{q\sqrt{q}}{q^{2}-1}}{\sqrt{q}}\Bigg)^{2}\leq\Bigg(\frac{1+\frac{q\sqrt{q}}{q^{2}-1}}{\sqrt{q}}\Bigg)^{2}
=\left(\frac{1+\tilde{\varepsilon}_{q}}{\sqrt{q}}\right)^{2}=\frac{1+\varepsilon_{q}}{q},
\end{align*}
where the two infinitesimals $\tilde{\varepsilon}_{q}$ and $\varepsilon_{q}$ are given by
\begin{equation*}
\tilde{\varepsilon}_{q}=\frac{q\sqrt{q}}{q^{2}-1},
\end{equation*}
\begin{equation*}
\varepsilon_{q}=\tilde{\varepsilon}_{q}^{2}+2\tilde{\varepsilon}_{q}
=\frac{q\sqrt{q}(2q^{2}+q\sqrt{q}-2)}{(q^{2}-1)^{2}}.
\end{equation*}
Besides, we have
\begin{equation*}
\left|\frac{1}{q+1}-\left(\frac{q}{q^{2}-1}\right)^{2}\right|=O(q^{-1}).
\end{equation*}

By {\bfseries Cases 1, 2} and {\bfseries 3}, we deduce that $\mathcal{M}$ satisfies the approximate symmetry condition of an $\varepsilon_{q}$-ASIC-POVM.

Finally, let us prove that $\mathcal{M}$ satisfies the informational completeness condition of an $\varepsilon_{q}$-ASIC-POVM.
It suffices to prove that $E_{1},E_{2},\ldots,E_{q^{2}}$ are linearly independent as elements of $M_{q}(\mathbb{C})$.

Suppose there exist $k_{a,b}\in\mathbb{C}$ and $m_{i}\in\mathbb{C}$ such that
\begin{equation}\label{eq2.15}
\sum_{a\in\mathbb{F}_{q}}\sum_{b\in\mathbb{F}_{q}^{\ast}}k_{a,b}E_{a,b}+\sum_{i=1}^{q}m_{i}U_{i}=0,
\end{equation}
where $E_{a,b}=\frac{1}{q}|v_{a,b}\rangle\langle v_{a,b}|$ for $a\in \mathbb{F}_{q}$ and $b\in\mathbb{F}_{q}^{\ast}$,
and $U_{i}=\frac{1}{q}|e_{i}\rangle\langle e_{i}|$ for $i=1,2,\ldots,q$.

Comparing the $(i,i)$-th entry of Eq. (\ref{eq2.15}), where $i=1,2,\ldots,q$, we get that
\begin{equation}\label{eq2.16}
\frac{1}{q^{2}}\sum_{a\in\mathbb{F}_{q}}\sum_{b\in\mathbb{F}_{q}^{\ast}}k_{a,b}+\frac{m_{i}}{q}=0.
\end{equation}

Comparing the $(i,j)$-th entry of Eq. (\ref{eq2.15}), where $1\leq i\neq j\leq q$, we have that
\begin{equation*}
\sum_{a\in \mathbb{F}_{q}}\sum_{b\in\mathbb{F}_{q}^{\ast}}k_{a,b}\chi\Big(a\big(f(a_{i})-f(a_{j})\big)+b(a_{i}-a_{j})\Big)=0.
\end{equation*}
In terms of Lemma \ref{lemma2.9}, we know that $k_{a,b}=0$ for all $a\in \mathbb{F}_{q}$ and $b\in\mathbb{F}_{q}^{\ast}$.
Inserting it into Eq. (\ref{eq2.16}), we have $m_{i}=0$ for all $i=1,2,\ldots,q$.
Hence, $E_{1},E_{2},\ldots,E_{q^{2}}$ are linearly independent as elements of $M_{q}(\mathbb{C})$,
and so are $M_{1},M_{2},\ldots,M_{q^{2}}$. Consequently, $\mathcal{M}$ satisfies the informational completeness condition of an $\varepsilon_{q}$-ASIC-POVM.

Therefore, we complete the proof. $\hfill\square$

\end{proof}

\begin{remark}\label{remark2.11}
Note that if we take $a_{i}=-a_{q+2-i}$ for $i=2,3,\ldots,\frac{q+1}{2}$, and put $a_{1}=0$ with $f(a_{1})=0$ in Theorem \ref{theorem2.10},
then we obtain the same $2$-to-$1$ PN functions $f(x)$ used in \cite[Theorem III.3]{Cao2017Two}.
Therefore, a natural question arises: whether there exist some $2$-to-$1$ PN functions that do not meet the condition in \cite[Theorem III.3]{Cao2017Two}.
To be specific, we wish to know whether there are some $2$-to-$1$ PN functions $f(x)$ that satisfy either $f(0)\neq 0$ or $f(x)\neq f(-x)$ for some $x\in \mathbb{F}_{q}^{\ast}$. We point out that the following polynomials over any finite field $\mathbb{F}_{q}$ of odd characteristic
are $2$-to-$1$ PN functions which meet the requirement.

\underline{1) The quadratic polynomials $sx^{2}+tx+w$ for all $sw\neq 0$.} This is because:
\begin{itemize}
\setlength{\itemsep}{1pt}
\setlength{\parsep}{1pt}
\setlength{\parskip}{1pt}
\item By \cite[p.7890]{Mesnager2019On}, all the quadratic polynomials over any odd characteristic finite field are $2$-to-$1$ mappings;

\item For any quadratic polynomial $f(x):=sx^{2}+tx+w$ with $s\neq 0$, it is not difficult to see that $f(x+a)-f(x)=2sax+sa^{2}+ta$ is bijective over any odd characteristic finite field for every $a\neq 0$, which reveals that all the quadratic polynomials over any odd characteristic finite field are PN functions.
\end{itemize}
\vskip -2mm

\underline{2) The polynomials $f(x)+d$ for all polynomials $f(x)$ from \cite[Table I]{Cao2017Two} and all $d\neq 0$.} This is because:
\begin{itemize}
\setlength{\itemsep}{1pt}
\setlength{\parsep}{1pt}
\setlength{\parskip}{1pt}
\item As $f(x)$ in \cite[Table I]{Cao2017Two} is a $2$-to-$1$ mapping, $f(x)+d$ with $d\neq 0$ is also a $2$-to-$1$ mapping;

\item Let $F(x):=f(x)+d$. As $f(x)$ in \cite[Table I]{Cao2017Two} is a PN function, one can see that $F(x+a)-F(x)=f(x+a)-f(x)$ is bijective for every $a\neq 0$, which reveals that $F(x)$ is also a PN function.
\end{itemize}
\end{remark}

\subsection{The biangular frame from Theorem \ref{theorem2.10}}\label{subsection2.3}

In this subsection, we will show that the set $\mathcal{A}$ from Theorem \ref{theorem2.10} forms a biangular frame \cite{Cahill2018Constructions,Magsino2019Biangular}.

\begin{definition}\label{definition2.13}
\rm{(\!\!\cite[pp.2]{Cahill2018Constructions})}
A set of vectors $\mathcal{F}=\{|f_{i}\rangle\}_{i=1}^{n}\subseteq \mathbb{C}^{d}$ is a \textbf{frame} if $\mathrm{span}\{|f_{i}\rangle\}_{i=1}^{n}=\mathbb{C}^{d}$. $\mathcal{F}$ is \textbf{unit-norm} if each frame vector $|f_{i}\rangle$ has norm $1$, i.e., it is normalized. For any unit-norm frame $\mathcal{F}=\{|f_{i}\rangle\}_{i=1}^{n}$, denote by $S_{\mathcal{F}}:=\{|\langle f_{i}|f_{i'}\rangle|:1\leq i\neq i'\leq n\}$ the set of \textbf{frame angles}.
$\mathcal{F}$ is called \textbf{$r$-angular} if $|S_{\mathcal{F}}|=r$ for some positive integer $r$.
In particular, $\mathcal{F}$ is called \textbf{equiangular} if $\mathcal{F}$ is $1$-angular; $\mathcal{F}$ is called \textbf{biangular} if $\mathcal{F}$ is $2$-angular.
\end{definition}

\begin{lemma}\label{lemma2.14}
Let $\{F_{i}\}_{i=1}^{n}$ be a collection of subnormalized projectors on $\mathbb{C}^{d}$ with $F_{i}=p_{i}|f_{i}\rangle\langle f_{i}|$,
where $p_{i}>0$ and $|f_{i}\rangle\in \mathbb{C}^{d}$ is a normalized vector for $i=1,\ldots,n$.
Denote $\mathcal{F}=\{|f_{i}\rangle\}_{i=1}^{n}$. If $\mathrm{span}\{F_{i}\}_{i=1}^{n}=M_{d}(\mathbb{C})$,
then $\mathcal{F}$ is a frame in $\mathbb{C}^{d}$.
\end{lemma}

\begin{proof}
By Definition \ref{definition2.13}, it suffices to prove that $\mathrm{span}\{|f_{i}\rangle\}_{i=1}^{n}=\mathbb{C}^{d}$.
For any $|w\rangle\in\mathbb{C}^{d}\backslash\{\mathbf{0}\}$, it follows from $\mathrm{span}\{F_{i}\}_{i=1}^{n}=M_{d}(\mathbb{C})$
that there exist $\{g_{i}\}_{i=1}^{n}\subseteq \mathbb{C}$ such that
\begin{equation*}
|w\rangle\langle w|=\sum_{i=1}^{n}g_{i}p_{i}|f_{i}\rangle\langle f_{i}|.
\end{equation*}
Then,
\begin{equation*}
|w\rangle=\frac{1}{\langle w|w\rangle}\sum_{i=1}^{n}\big(g_{i}p_{i}\langle f_{i}|w\rangle\big)|f_{i}\rangle\in\mathrm{span}\{|f_{i}\rangle\}_{i=1}^{n}.
\end{equation*}
Hence $\mathbb{C}^{d}\backslash\{\mathbf{0}\}\subseteq \mathrm{span}\{|f_{i}\rangle\}_{i=1}^{n}$.
This together with $\mathbf{0}\in \mathrm{span}\{|f_{i}\rangle\}_{i=1}^{n}$ produces that $\mathbb{C}^{d}\subseteq \mathrm{span}\{|f_{i}\rangle\}_{i=1}^{n}$,
and thus $\mathbb{C}^{d}=\mathrm{span}\{|f_{i}\rangle\}_{i=1}^{n}$. This completes the proof. $\hfill\square$
\end{proof}

\vspace{4pt}

By Lemma \ref{lemma2.14}, we show that the set $\mathcal{A}$ in Eq. (\ref{eq2.12}) forms a biangular frame in the following theorem.

\begin{theorem}\label{theorem2.15}
With the notations in Theorem \ref{theorem2.10}. Then the set $\mathcal{A}$ in Eq. (\ref{eq2.12}) forms a biangular frame in $\mathbb{C}^{q}$.
\end{theorem}

\begin{proof}
By Theorem \ref{theorem2.10}, $\mathrm{span}\{E_{i}\}_{i=1}^{n}=M_{d}(\mathbb{C})$. Combining it with Lemma \ref{lemma2.14},
we know that $\mathcal{A}$ is a frame in $\mathbb{C}^{q}$.
Further, let us prove that it is biangular by the following cases.

{\bfseries Case 1:} When $a\neq c$, it follows from Lemma \ref{lemma2.8} that
$\big|\langle v_{a,b}|v_{c,d}\rangle\big|=|\frac{1}{q}\sum_{i=1}^{q}\chi\big((c-a)f(a_{i})+(d-b)a_{i}\big)|=\frac{1}{\sqrt{q}}$.

{\bfseries Case 2:} When $a=c$ and $b\neq d$, $|\langle v_{a,b}|v_{c,d}\rangle|=|\frac{1}{q}\sum_{i=1}^{q}\chi\big((d-b)a_{i}\big)|=0$.

{\bfseries Case 3:} Note that $|\langle e_{i}|v_{a,b}\rangle|=|\frac{1}{\sqrt{q}}\chi\big(af(a_{i})+ba_{i}\big)|=\frac{1}{\sqrt{q}}$.

{\bfseries Case 4:} When $i\neq j$, $|\langle e_{i}|e_{j}\rangle|=0$.

Combing these four cases, we get that the cardinality of the frame angles of $\mathcal{A}$ is $|S_{\mathcal{A}}|=2$.
Therefore, $\mathcal{A}$ is biangular in $\mathbb{C}^{q}$, which completes the proof. $\hfill\square$
\end{proof}

\subsection{Comparison of our results with known ones}\label{subsection2.5}

By using $2$-to-$1$ PN functions, Theorem \ref{theorem2.10} and \cite[Theorem III.3]{Cao2017Two} construct $\varepsilon_{q}$-ASIC-POVMs.
Although their parameters (the dimension and cardinality) are the same, we summarize other differences of these $\varepsilon_{q}$-ASIC-POVMs and their corresponding biangular frames in the following.
\vspace{-3pt}
\begin{itemize}
\item [(1)] In \cite[Theorem III.3]{Cao2017Two}, the authors constructed the $\varepsilon_{q}$-ASIC-POVMs by using a special kind of $2$-to-$1$ PN functions $f(x)$ satisfying $f(0)=0$ and $f(x)=f(y)$ iff $x=-y$ for all $x,y\in\mathbb{F}_{q}$.
    In Theorem \ref{theorem2.10}, we construct a more general class of $\varepsilon_{q}$-ASIC-POVMs from arbitrary $2$-to-$1$ PN functions.
\vspace{-3pt}

\item [(2)] To guarantee that the construction satisfies the informational completeness condition,
\cite[Theorem III.3]{Cao2017Two} utilizes certain properties of mutually unbiased bases (MUBs),
while Theorem \ref{theorem2.10} utilizes Lemmas \ref{lemma2.1} and \ref{lemma2.9}.

\vspace{-3pt}

\item [(3)] \underline{Practical application}. As shown in Theorem \ref{theorem2.15}, the set of vectors from Theorem \ref{theorem2.10}, i.e.,
$\mathcal{A}=\{|v_{a,b}\rangle:a\in\mathbb{F}_{q},\ b\in\mathbb{F}_{q}^{\ast}\}\cup \{|e_{1}\rangle\cup\ldots\cup|e_{q}\rangle\}$ with $|v_{a,b}\rangle=\frac{1}{\sqrt{q}}\big(\chi(af(a_{i})+ba_{i})\big)^{T}_{i=1,\ldots,q}$, forms a biangular frame.
Comparing to the biangular frames derived from \cite[Theorem III.3]{Cao2017Two}, we can obtain more kinds of biangular frames from Theorem \ref{theorem2.10} since the $2$-to-$1$ PN functions $f(x)$ used in Theorem \ref{theorem2.10} are general.
One approach to Zauner's conjecture, and more generally to constructions of equiangular tight frames (ETFs), is to find continuous families of biangular frames that contain SIC-POVMs or ETFs (see, e.g., \cite{Cahill2018Constructions,Magsino2019Biangular}).
From this perspective, it makes sense to obtain more kinds of biangular frames.

\vspace{-3pt}

\item [(4)] \underline{Potential application}. For other possible applications of $\varepsilon_{q}$-ASIC-POVMs in the future
(e.g., they might be useful in quantum cryptography \cite{Boileau2005Unconditional},
dynamic quantum tomography \cite{Czerwinski2021Quantum}, etc),
we characterize ``how close'' the $\varepsilon_{q}$-ASIC-POVMs are from being SIC-POVMs of dimension $q$. Specifically,
i) we give the explicit expressions about all $\varepsilon_{q}$ in different kinds of upper bounds of $q^{2}\mathrm{Tr}(M_{i}M_{j})$ in Theorem \ref{theorem2.10};
ii) we calculate the absolute value of the difference
$|\frac{1}{q+1}-A|$ in Theorem \ref{theorem2.10}, where
$\frac{1}{q+1}$ is the value of $d^{2}\mathrm{Tr}(E_{i}E_{j})$ in Definition \ref{definition1.3} (ii) for $d=q$, and
$A$ is the explicit upper bound of $d^{2}\mathrm{Tr}(A_{i}A_{j})$ in Definition \ref{definition1.4} (ii) for $d=q$.
\end{itemize}

\section{On $\varepsilon_{q+1}$-ASIC-POVMs and corresponding set of vectors via the Li bound}\label{section3}

Li's character sum bound (the Li bound, for short), obtained by Wen-Ch'ing Li \cite{Li1992Character}, is an estimate on a class of multiplicative character sums. In this section, we will present the construction of $\varepsilon_{q+1}$-ASIC-POVMs by using the Li bound.
We will also show that the set of vectors corresponding to the $\varepsilon_{q+1}$-ASIC-POVM forms a $\left((q+1)^{2},q+1\right)$ asymptotically optimal codebook.

\subsection{Notation}

Let $q$ be a prime power. We also define
\begin{itemize}
\item $\mathbb{F}_{q^{3}}^{\ast}=\langle\alpha\rangle$.

\item $\mathbb{F}_{q}=\{b_{1},b_{2},\ldots,b_{q}\}$.

\item $N=\{x\in\mathbb{F}_{q^{3}}^{\ast}|\mathrm{Nr}: \mathbb{F}_{q^{3}}^{\ast}\rightarrow \mathbb{F}_{q}^{\ast},\ \mathrm{Nr}(x)=1\}$,
where $\mathrm{Nr}$ denotes the norm mapping from $\mathbb{F}_{q^{3}}^{\ast}$ to $\mathbb{F}_{q}^{\ast}$.

\item $S=\{(\alpha-b)^{q-1}|b\in\mathbb{F}_{q}\}\cup\{1\}$.
\end{itemize}

\subsection{Construction of $\varepsilon_{q+1}$-ASIC-POVMs}

Let us give the following two lemmas.

\begin{lemma}\label{lemma3.1}
$N=\{\alpha^{m(q-1)}|\ m=1,2,\ldots,q^{2}+q+1\}$.
\end{lemma}

\begin{proof}
By the definition of $N$, we know that $\mathrm{Nr}(a)=1\Longleftrightarrow a^{\frac{q^{3}-1}{q-1}}=1$.
Hence, the elements in $N$ are exactly the roots of $x^{\frac{q^{3}-1}{q-1}}=1$, which completes the proof. $\hfill\square$
\end{proof}

\begin{lemma}\label{lemma3.2}
Let $S=\{d_{1},\ldots,d_{q},d_{q+1}\}$, where $d_{i}=(\alpha-b_{i})^{q-1}$ for $i=1,2,\ldots,q$, and $d_{q+1}=1$.
Define
\begin{equation*}
T=\{d_{i}d_{j}^{-1}|\ 1\leq i\neq j\leq q+1\}.
\end{equation*}
Then, $T=\{\alpha^{m(q-1)}|\ m=1,2,\ldots,q^{2}+q\}$.
\end{lemma}

\begin{proof}
Let us prove that for any $d_{i_{1}}d_{j_{1}}^{-1},d_{i_{2}}d_{j_{2}}^{-1}\in T$, we have
$d_{i_{1}}d_{j_{1}}^{-1}\neq d_{i_{2}}d_{j_{2}}^{-1}$ if $(i_{1},j_{1})\neq(i_{2},j_{2})$. Otherwise, we have the following cases.

{\bfseries Case 1:} Suppose $d_{i_{1}}d_{j_{1}}^{-1}=d_{i_{2}}d_{j_{2}}^{-1}\in T$,
where $1\leq i_{1},i_{2},j_{1},j_{2}\leq q$ and  $(i_{1},j_{1})\neq(i_{2},j_{2})$. Then,
\begin{equation*}
\bigg[\frac{(\alpha-b_{i_{1}})(\alpha-b_{j_{2}})}{(\alpha-b_{j_{1}})(\alpha-b_{i_{2}})}\bigg]^{q-1}=1,
\end{equation*}
and thus $\frac{(\alpha-b_{i_{1}})(\alpha-b_{j_{2}})}{(\alpha-b_{j_{1}})(\alpha-b_{i_{2}})}\in\mathbb{F}_{q}^{\ast}$.
We denote this ratio by $c$. Then, $(\alpha-b_{i_{1}})(\alpha-b_{j_{2}})-c(\alpha-b_{j_{1}})(\alpha-b_{i_{2}})=0$.
In this case, if $c\neq 1$, we deduce that $\alpha$ is a root of a quadratic polynomial over $\mathbb{F}_{q}$.
However, the degree of the minimal polynomial of $\alpha$ over $\mathbb{F}_{q}$ is $3$, which leads to a contradiction.
Hence, $c=1$, and we obtain that
\begin{equation*}
[(b_{i_{1}}+b_{j_{2}})-(b_{i_{2}}+b_{j_{1}})]\alpha+b_{i_{2}}b_{j_{1}}-b_{i_{1}}b_{j_{2}}=0,
\end{equation*}
which implies that
\begin{equation*}
b_{i_{1}}+b_{j_{2}}=b_{i_{2}}+b_{j_{1}},\ b_{i_{1}}b_{j_{2}}=b_{i_{2}}b_{j_{1}}.
\end{equation*}

Then, we have
\begin{equation*}
b_{i_{1}}=b_{j_{1}},\ b_{i_{2}}=b_{j_{2}}
\end{equation*}
or
\begin{equation*}
b_{i_{1}}=b_{i_{2}},\ b_{j_{1}}=b_{j_{2}}.
\end{equation*}
The first case illustrates that $d_{i_{1}}=d_{j_{1}}$, and thus $d_{i_{1}}d_{j_{1}}^{-1}=1\in T$, which leads to a contradiction.
The second case means that $(i_{1},j_{1})=(i_{2},j_{2})$, which is also a contradiction.

{\bfseries Case 2:} Suppose $d_{i_{1}}d_{j_{1}}^{-1}=d_{i_{2}}d_{q+1}^{-1}\in T$, where $1\leq i_{1},j_{1},i_{2}\leq q$. Then,
\begin{equation*}
\bigg[\frac{\alpha-b_{i_{1}}}{(\alpha-b_{j_{1}})(\alpha-b_{i_{2}})}\bigg]^{q-1}=1.
\end{equation*}
Similar to Case 1, we deduce that $\alpha$ is a root of a quadratic polynomial over $\mathbb{F}_{q}$, which is a contradiction.

{\bfseries Case 3:} Suppose $d_{i_{1}}d_{j_{1}}^{-1}=d_{q+1}d_{j_{2}}^{-1}\in T$, where $1\leq i_{1},j_{1},j_{2}\leq q$.
Then, similar to Case 1, we obtain a contradiction.

{\bfseries Case 4:} Suppose $d_{i_{1}}d_{q+1}^{-1}=d_{i_{2}}d_{q+1}^{-1}\in T$, where $1\leq i_{1}\neq i_{2}\leq q$.
Then, $d_{i_{1}}=d_{i_{2}}$ and thus $i_{1}=i_{2}$, which is a contradiction.

{\bfseries Case 5:} Suppose $d_{q+1}d_{j_{1}}^{-1}=d_{q+1}d_{j_{2}}^{-1}\in T$, where $1\leq j_{1}\neq j_{2}\leq q$.
Then, $d_{j_{1}}=d_{j_{2}}$ and thus $j_{1}=j_{2}$, which is a contradiction.

By the above analysis, we obtain $\sharp T=q(q+1)=q^{2}+q$. Moreover, it is verified that
$T\subseteq\{\alpha^{m(q-1)}|\ m=1,2,\ldots,q^{2}+q\}$.
Therefore, $T=\{\alpha^{m(q-1)}|\ m=1,2,\ldots,q^{2}+q\}$, which completes the proof. $\hfill\square$
\end{proof}

\vspace{6pt}

Based on Lemmas \ref{lemma3.1} and \ref{lemma3.2}, we obtain the following lemma.

\begin{lemma}\label{lemma3.3}
Let $q$ be a prime power. Let $S=\{d_{1},\ldots,d_{q},d_{q+1}\}$, where $d_{i}=(\alpha-b_{i})^{q-1}$ for $i=1,2,\ldots,q$, and $d_{q+1}=1$.
Denote by $\widehat{N}$ and $\widehat{N}\backslash \{1\}$ the multiplicative character group of $N$ and the set consisting of
all nontrivial multiplicative characters of $N$, respectively.
If $k_{\psi}\in\mathbb{C}$ (for $\psi\in\widehat{N}\backslash \{1\}$) such that
\begin{equation}\label{eq3.1}
\sum_{\psi\in\widehat{N}\backslash \{1\}}k_{\psi}\psi(d_{i}d_{j}^{-1})=0
\end{equation}
for all $1\leq i\neq j\leq q+1$,
then $k_{\psi}\equiv0$ for all $\psi\in\widehat{N}\backslash \{1\}$.
\end{lemma}

\begin{proof}
By Lemma \ref{lemma3.1}, we have that
\begin{equation*}
\widehat{N}\backslash \{1\}=\{\varphi_{m}|\ m=1,2,\ldots,q^{2}+q\}
\end{equation*}
with $\varphi_{m}(\alpha^{j(q-1)})=\zeta_{q^{2}+q+1}^{mj}$ for $j=1,2,\ldots,q^{2}+q+1$,
where $\zeta_{q^{2}+q+1}$ denotes a primitive $(q^{2}+q+1)$-th root of unity in $\mathbb{C}$.
By Lemma \ref{lemma3.2}, we know Eq. (\ref{eq3.1}) is equivalent to
\begin{equation*}
\sum_{m=1}^{q^{2}+q}k_{m}\varphi_{m}(\alpha^{n(q-1)})=0
\end{equation*}
for all $n=1,2,\ldots,q^{2}+q$. That is,
\begin{equation*}
\sum_{m=1}^{q^{2}+q}k_{m}\zeta_{q^{2}+q+1}^{mn}=0
\end{equation*}
for all $n=1,2,\ldots,q^{2}+q$.

Let us show that $k_{m}\equiv0$ for all $m=1,2,\ldots,q^{2}+q$. To this end, consider the following polynomial
\begin{equation*}
g(x)=\sum_{m=1}^{q^{2}+q}k_{m}x^{m}.
\end{equation*}
Then, $R:=\big\{\zeta_{q^{2}+q+1}^{n}\big|\ n=1,2,\ldots,q^{2}+q\big\}$ are $q^{2}+q$ distinct roots of $g(x)$.
Since $g(0)=0$, we deduce that $g(x)$ has at least $q^{2}+q+1$ distinct roots.
By $\mathrm{deg}(g(x))=q^{2}+q$, we have $g(x)=0$, which derives $k_{m}\equiv0$ for all $m=1,2,\ldots,q^{2}+q$.
This completes the proof. $\hfill\square$
\end{proof}

\vspace{6pt}

The following lemma corresponds to the case $n=3$ in \cite[Theorem 6]{Li1992Character}, which is called the Li bound for short,
and will be used later in the construction of $\varepsilon_{q+1}$-ASIC-POVMs.

\begin{lemma}\label{lemma3.4}
\rm{(\!\!\cite[Theorem 6]{Li1992Character}})
\emph{(The Li bound) Let $q$ be a prime power, and let $N$, $\widehat{N}$, and $S$ be as in Lemma \ref{lemma3.3}.
Then, for each $\psi\in\widehat{N}\backslash \{1\}$,
\begin{equation*}
\big|\sum_{s\in S}\psi(s)\big|\leq\sqrt{q}.
\end{equation*}
}
\end{lemma}

By Lemmas \ref{lemma3.3} and \ref{lemma3.4}, we are ready to present the following construction of $\varepsilon_{q+1}$-ASIC-POVMs.

\begin{theorem}\label{theorem3.5}
Let $q$ be a prime power, and let $N$, $\widehat{N}$, and $S$ be as in Lemma \ref{lemma3.3}.
Define
\begin{equation*}
\mathcal{C}=\big\{|u_{\psi}\rangle:\psi\in\widehat{N}\backslash \{1\}\big\}\cup \{|e_{1}\rangle\cup\ldots\cup|e_{q+1}\rangle\},
\end{equation*}
where $|u_{\psi}\rangle=\frac{1}{\sqrt{q+1}}(\psi(d_{1}),\psi(d_{2}),\ldots,\psi(d_{q+1}))^{T}$ for each $\psi\in\widehat{N}\backslash\{1\}$, and
$|e_{i}\rangle\in\mathbb{C}^{q+1}$ is the unit vector whose $i$-th component is $1$, and $0$ in other components.
Define $E=\sum_{i=1}^{(q+1)^{2}}E_{i}$, where $E_{i}=\frac{1}{q+1}|u_{i}\rangle\langle u_{i}|$ for $|u_{i}\rangle\in \mathcal{C}$.
Then $\mathcal{F}=\{F_{i}:=E^{-\frac{1}{2}}E_{i}E^{-\frac{1}{2}}|\ i=1,2,\ldots,(q+1)^{2}\}$ forms an $\varepsilon_{q+1}$-ASIC-POVM, where

1) For $F_{i}=\frac{1}{q+1}E^{-\frac{1}{2}}|u_{\tau}\rangle\langle u_{\tau}|E^{-\frac{1}{2}}$ and
$F_{j}=\frac{1}{q+1}E^{-\frac{1}{2}}|u_{\psi}\rangle\langle u_{\psi}|E^{-\frac{1}{2}}$ with $\tau\neq\psi\in\widehat{N}\backslash \{1\}$,
\begin{equation*}
(q+1)^{2}\mathrm{Tr}(F_{i}F_{j})\leq \left(\frac{\sqrt{q}(q+1)(q^{2}+q+\sqrt{q}+1)}{(q^{2}+2q+2)(q^{2}+q+1)}\right)^{2}=\frac{1+\varepsilon_{q+1}}{q+1},
\end{equation*}
where the infinitesimal $\varepsilon_{q+1}$ is given by
\begin{equation*}
\varepsilon_{q+1}=\frac{q(q+1)^{3}(q^{2}+q+\sqrt{q}+1)^{2}-(q^{2}+2q+2)^{2}(q^{2}+q+1)^{2}}{(q^{2}+2q+2)^{2}(q^{2}+q+1)^{2}}.
\end{equation*}
Besides, we have
\begin{equation*}
\left|\frac{1}{q+2}-\left(\frac{\sqrt{q}(q+1)(q^{2}+q+\sqrt{q}+1)}{(q^{2}+2q+2)(q^{2}+q+1)}\right)^{2}\right|=O(q^{-\frac{5}{2}}).
\end{equation*}

\vspace{6pt}

2) For $F_{i}=\frac{1}{q+1}E^{-\frac{1}{2}}|u_{\tau}\rangle\langle u_{\tau}|E^{-\frac{1}{2}}$ and
$F_{j}=\frac{1}{q+1}E^{-\frac{1}{2}}|e_{j}\rangle\langle e_{j}|E^{-\frac{1}{2}}$ with $\tau\in\widehat{N}\backslash \{1\}$ and $j=1,2,\ldots,q+1$,
\begin{equation*}
(q+1)^{2}\mathrm{Tr}(F_{i}F_{j})\leq\left(\frac{(q+1)^{\frac{3}{2}}(q^{2}+q+\sqrt{q}+1)}{(q^{2}+2q+2)(q^{2}+q+1)}\right)^{2}
=\frac{1+\varepsilon_{q+1}}{q+1},
\end{equation*}
where the infinitesimal $\varepsilon_{q+1}$ is given by
\begin{equation*}
\varepsilon_{q+1}=\frac{(q+1)^{4}(q^{2}+q+\sqrt{q}+1)^{2}-(q^{2}+2q+2)^{2}(q^{2}+q+1)^{2}}{(q^{2}+2q+2)^{2}(q^{2}+q+1)^{2}}.
\end{equation*}
Besides, we have
\begin{equation*}
\left|\frac{1}{q+2}-\left(\frac{(q+1)^{\frac{3}{2}}(q^{2}+q+\sqrt{q}+1)}{(q^{2}+2q+2)(q^{2}+q+1)}\right)^{2}\right|=O(q^{-2}).
\end{equation*}

\vspace{6pt}

3) For $F_{i}=\frac{1}{q+1}E^{-\frac{1}{2}}|e_{i}\rangle\langle e_{i}|E^{-\frac{1}{2}}$ and
$F_{j}=\frac{1}{q+1}E^{-\frac{1}{2}}|e_{j}\rangle\langle e_{j}|E^{-\frac{1}{2}}$ with $1\leq i\neq j\leq q+1$,
\begin{equation*}
(q+1)^{2}\mathrm{Tr}(F_{i}F_{j})=\frac{(q+1)^{4}}{(q^{2}+2q+2)^{2}(q^{2}+q+1)^{2}}\leq\frac{1+\varepsilon_{q+1}}{q+1},
\end{equation*}
where the infinitesimal $\varepsilon_{q+1}$ is given by
\begin{equation*}
\varepsilon_{q+1}=\frac{(q+1)^{5}+2(q+1)^{\frac{5}{2}}(q^{2}+2q+2)(q^{2}+q+1)}{(q^{2}+2q+2)^{2}(q^{2}+q+1)^{2}}.
\end{equation*}
Besides, we have
\begin{equation*}
\left|\frac{1}{q+2}-\frac{(q+1)^{4}}{(q^{2}+2q+2)^{2}(q^{2}+q+1)^{2}}\right|=O(q^{-1}).
\end{equation*}
\end{theorem}

\begin{proof}
First, let us show that $\mathcal{F}$ satisfies the completeness/POVM condition of an $\varepsilon_{q+1}$-ASIC-POVM.
We have
\begin{equation*}
E=\sum_{i=1}^{(q+1)^{2}}E_{i}=\frac{1}{q+1}\sum_{\psi\in\widehat{N}\backslash \{1\}}|u_{\psi}\rangle\langle u_{\psi}|+\frac{1}{q+1}I_{q+1}.
\end{equation*}

For any $1\leq i,j\leq q+1$, the $(i,j)$-th entry of $E$ is
\begin{equation*}
E(i,j)=\frac{1}{(q+1)^{2}}\sum_{\psi\in\widehat{N}\backslash \{1\}}\psi(d_{i}d_{j}^{-1})+\frac{1}{q+1}\delta_{i,j}.
\end{equation*}
Then, we see that

(a) For $i=1,2,\ldots,q+1$, we have $E(i,i)=\frac{1}{(q+1)^{2}}\cdot (q^{2}+q)+\frac{1}{q+1}=1$.

(b) For $1\leq i\neq j\leq q+1$, we have $E(i,j)=-\frac{1}{(q+1)^{2}}$.

Hence, we obtain that
\begin{equation*}
E=\frac{(q+1)^{2}+1}{(q+1)^{2}}I_{q+1}-\frac{1}{(q+1)^{2}}V_{q+1},
\end{equation*}
where $V_{q+1}$ is an $(q+1)\times (q+1)$ matrix whose entries are all $1$. Then, we calculate that
\begin{equation*}
E^{-1}=\frac{(q+1)^{2}}{q^{2}+2q+2}I_{q+1}+\frac{(q+1)^{2}}{(q^{2}+2q+2)(q^{2}+q+1)}V_{q+1}.
\end{equation*}

For each $u=1,2,\ldots,q+1$, we compute that the $u$-th leading principal minor of $E$ is
\begin{equation*}
\Big(1-\frac{u}{(q+1)^{2}+1}\Big)\Big(1+\frac{1}{(q+1)^{2}}\Big)^{u}>0,
\end{equation*}
which illustrates that $E$ is positive definite.
Hence, $E^{-1}$ exists and it is positive definite as well, which implies that we can obtain the uniquely determined positive definite matrix $E^{-\frac{1}{2}}$.
Consequently,
\begin{equation*}
\sum_{i=1}^{(q+1)^{2}}F_{i}=E^{-\frac{1}{2}}\Bigg(\sum_{i=1}^{(q+1)^{2}}E_{i}\Bigg)E^{-\frac{1}{2}}=I_{q+1}.
\end{equation*}
Therefore, $\mathcal{F}$ satisfies the completeness/POVM condition of an $\varepsilon_{q+1}$-ASIC-POVM.

Next, let us prove that $\mathcal{F}$ satisfies the approximate symmetry condition of an $\varepsilon_{q+1}$-ASIC-POVM.

{\bfseries Case 1:} For $F_{i}=\frac{1}{q+1}E^{-\frac{1}{2}}|u_{\tau}\rangle\langle u_{\tau}|E^{-\frac{1}{2}}$ and
$F_{j}=\frac{1}{q+1}E^{-\frac{1}{2}}|u_{\psi}\rangle\langle u_{\psi}|E^{-\frac{1}{2}}$,
where $\tau\neq\psi\in\widehat{N}\backslash \{1\}$, we obtain
\begin{align*}
(q+1)^{2}\mathrm{Tr}(F_{i}F_{j})
&=|\langle u_{\tau}|E^{-1}|u_{\psi}\rangle|^{2}\\
&=\Big|\frac{(q+1)^{2}}{q^{2}+2q+2}\langle u_{\tau}|u_{\psi}\rangle+\frac{(q+1)^{2}}{(q^{2}+2q+2)(q^{2}+q+1)}\langle u_{\tau}|V_{q+1}|u_{\psi}\rangle\Big|^{2}.
\end{align*}

It follows from Lemma \ref{lemma3.4} that
\begin{equation*}
\big|\langle u_{\tau}|u_{\psi}\rangle\big|=\frac{1}{q+1}\Big|\sum_{i=1}^{q+1}(\tau^{-1}\psi)(d_{i})\Big|\leq\frac{\sqrt{q}}{q+1},
\end{equation*}
\begin{equation*}
\big|\langle u_{\tau}|V_{q+1}|u_{\psi}\rangle\big|=\frac{1}{q+1}\Big|\sum_{i=1}^{q+1}\tau^{-1}(d_{i})\Big|\cdot\Big|\sum_{i=1}^{q+1}\psi(d_{i})\Big|         \leq\frac{q}{q+1}.
\end{equation*}
Hence, we have that
\begin{align*}
(q+1)^{2}\mathrm{Tr}(F_{i}F_{j})
&\leq \bigg(\frac{(q+1)^{2}}{q^{2}+2q+2}\cdot \frac{\sqrt{q}}{q+1}+\frac{(q+1)^{2}}{(q^{2}+2q+2)(q^{2}+q+1)}\cdot \frac{q}{q+1}\bigg)^{2}\\
&=\left(\frac{\sqrt{q}(q+1)(q^{2}+q+\sqrt{q}+1)}{(q^{2}+2q+2)(q^{2}+q+1)}\right)^{2}\\
&=\left(\frac{\frac{\sqrt{q}(q+1)^{\frac{3}{2}}}{q^{2}+2q+2}+\frac{q(q+1)^{\frac{3}{2}}}{(q^{2}+2q+2)(q^{2}+q+1)}}{\sqrt{q+1}}\right)^{2}
=\left(\frac{1+\tilde{\varepsilon}_{q+1}}{\sqrt{q+1}}\right)^{2}=\frac{1+\varepsilon_{q+1}}{q+1},
\end{align*}
where the two infinitesimals $\tilde{\varepsilon}_{q+1}$ and $\varepsilon_{q+1}$ are given by
\begin{equation*}
\tilde{\varepsilon}_{q+1}=\frac{\sqrt{q}(q+1)^{\frac{3}{2}}(q^{2}+q+\sqrt{q}+1)-(q^{2}+2q+2)(q^{2}+q+1)}{(q^{2}+2q+2)(q^{2}+q+1)},
\end{equation*}
\begin{equation*}
\varepsilon_{q+1}=\tilde{\varepsilon}_{q+1}^{2}+2\tilde{\varepsilon}_{q+1}
=\frac{q(q+1)^{3}(q^{2}+q+\sqrt{q}+1)^{2}-(q^{2}+2q+2)^{2}(q^{2}+q+1)^{2}}{(q^{2}+2q+2)^{2}(q^{2}+q+1)^{2}}.
\end{equation*}
Besides, we have
\begin{equation*}
\left|\frac{1}{q+2}-\left(\frac{\sqrt{q}(q+1)(q^{2}+q+\sqrt{q}+1)}{(q^{2}+2q+2)(q^{2}+q+1)}\right)^{2}\right|=O(q^{-\frac{5}{2}}).
\end{equation*}

{\bfseries Case 2:} For $F_{i}=\frac{1}{q+1}E^{-\frac{1}{2}}|u_{\tau}\rangle\langle u_{\tau}|E^{-\frac{1}{2}}$ and
$F_{j}=\frac{1}{q+1}E^{-\frac{1}{2}}|e_{j}\rangle\langle e_{j}|E^{-\frac{1}{2}}$, where $\tau\in\widehat{N}\backslash \{1\}$ and $j=1,2,\ldots,q+1$,
we obtain
\begin{align*}
(q+1)^{2}\mathrm{Tr}(F_{i}F_{j})
&=|\langle u_{\tau}|E^{-1}|e_{j}\rangle|^{2}\\
&=\Bigg|\frac{(q+1)^{2}}{q^{2}+2q+2}\langle u_{\tau}|e_{j}\rangle+\frac{(q+1)^{2}}{(q^{2}+2q+2)(q^{2}+q+1)}\langle u_{\tau}|V_{q+1}|e_{j}\rangle\Bigg|^{2}\\
&=\Bigg|\frac{(q+1)^{2}}{q^{2}+2q+2}\cdot\frac{1}{\sqrt{q+1}}\tau^{-1}(d_{j})
+\frac{(q+1)^{2}}{(q^{2}+2q+2)(q^{2}+q+1)}\cdot\frac{1}{\sqrt{q+1}}\sum_{i=1}^{q+1}\tau^{-1}(d_{i})\Bigg|^{2}\\
&\leq\Bigg(\frac{(q+1)^{2}}{q^{2}+2q+2}\cdot\frac{1}{\sqrt{q+1}}
+\frac{(q+1)^{2}}{(q^{2}+2q+2)(q^{2}+q+1)}\cdot\frac{1}{\sqrt{q+1}}\cdot\sqrt{q}\Bigg)^{2}\\
&=\left(\frac{(q+1)^{\frac{3}{2}}(q^{2}+q+\sqrt{q}+1)}{(q^{2}+2q+2)(q^{2}+q+1)}\right)^{2}\\
&=\Bigg(\frac{\frac{(q+1)^{2}}{q^{2}+2q+2}+\frac{\sqrt{q}(q+1)^{2}}{(q^{2}+2q+2)(q^{2}+q+1)}}{\sqrt{q+1}}\Bigg)^{2}
=\Bigg(\frac{1+\tilde{\varepsilon}_{q+1}}{\sqrt{q+1}}\Bigg)^{2}=\frac{1+\varepsilon_{q+1}}{q+1},
\end{align*}
where the two infinitesimals $\tilde{\varepsilon}_{q+1}$ and $\varepsilon_{q+1}$ are given by
\begin{equation*}
\tilde{\varepsilon}_{q+1}=\frac{(q+1)^{2}(q^{2}+q+\sqrt{q}+1)-(q^{2}+2q+2)(q^{2}+q+1)}{(q^{2}+2q+2)(q^{2}+q+1)},
\end{equation*}
\begin{equation*}
\varepsilon_{q+1}=\tilde{\varepsilon}_{q+1}^{2}+2\tilde{\varepsilon}_{q+1}
=\frac{(q+1)^{4}(q^{2}+q+\sqrt{q}+1)^{2}-(q^{2}+2q+2)^{2}(q^{2}+q+1)^{2}}{(q^{2}+2q+2)^{2}(q^{2}+q+1)^{2}}.
\end{equation*}
Besides, we have
\begin{equation*}
\left|\frac{1}{q+2}-\left(\frac{(q+1)^{\frac{3}{2}}(q^{2}+q+\sqrt{q}+1)}{(q^{2}+2q+2)(q^{2}+q+1)}\right)^{2}\right|=O(q^{-2}).
\end{equation*}

{\bfseries Case 3:} For $F_{i}=\frac{1}{q+1}E^{-\frac{1}{2}}|e_{i}\rangle\langle e_{i}|E^{-\frac{1}{2}}$ and
$F_{j}=\frac{1}{q+1}E^{-\frac{1}{2}}|e_{j}\rangle\langle e_{j}|E^{-\frac{1}{2}}$, where $1\leq i\neq j\leq q+1$, we have
\begin{align*}
(q+1)^{2}\mathrm{Tr}(F_{i}F_{j})
&=|\langle e_{i}|E^{-1}|e_{j}\rangle|^{2}\\
&=\Bigg|\frac{(q+1)^{2}}{q^{2}+2q+2}\langle e_{i}|e_{j}\rangle+\frac{(q+1)^{2}}{(q^{2}+2q+2)(q^{2}+q+1)}\langle e_{i}|V_{q+1}|e_{j}\rangle\Bigg|^{2}\\
&=\Bigg|\frac{(q+1)^{2}}{q^{2}+2q+2}\cdot 0+\frac{(q+1)^{2}}{(q^{2}+2q+2)(q^{2}+q+1)}\cdot 1\Bigg|^{2}\\
&=\frac{(q+1)^{4}}{(q^{2}+2q+2)^{2}(q^{2}+q+1)^{2}}\\
&=\Bigg(\frac{\frac{(q+1)^{\frac{5}{2}}}{(q^{2}+2q+2)(q^{2}+q+1)}}{\sqrt{q+1}}\Bigg)^{2}\\
&\leq\Bigg(\frac{1+\frac{(q+1)^{\frac{5}{2}}}{(q^{2}+2q+2)(q^{2}+q+1)}}{\sqrt{q+1}}\Bigg)^{2}
=\Bigg(\frac{1+\tilde{\varepsilon}_{q+1}}{\sqrt{q+1}}\Bigg)^{2}=\frac{1+\varepsilon_{q+1}}{q+1},
\end{align*}
where the two infinitesimals $\tilde{\varepsilon}_{q+1}$ and $\varepsilon_{q+1}$ are given by
\begin{equation*}
\tilde{\varepsilon}_{q+1}=\frac{(q+1)^{\frac{5}{2}}}{(q^{2}+2q+2)(q^{2}+q+1)},
\end{equation*}
\begin{equation*}
\varepsilon_{q+1}=\tilde{\varepsilon}_{q+1}^{2}+2\tilde{\varepsilon}_{q+1}
=\frac{(q+1)^{5}+2(q+1)^{\frac{5}{2}}(q^{2}+2q+2)(q^{2}+q+1)}{(q^{2}+2q+2)^{2}(q^{2}+q+1)^{2}}.
\end{equation*}
Besides, we have
\begin{equation*}
\left|\frac{1}{q+2}-\frac{(q+1)^{4}}{(q^{2}+2q+2)^{2}(q^{2}+q+1)^{2}}\right|=O(q^{-1}).
\end{equation*}

By {\bfseries Cases 1, 2} and {\bfseries 3}, we deduce that $\mathcal{F}$ satisfies the approximate symmetry condition of an $\varepsilon_{q+1}$-ASIC-POVM.

Finally, let us prove that $\mathcal{F}$ satisfies the informational completeness condition of an $\varepsilon_{q+1}$-ASIC-POVM.
It suffices to prove that $E_{1},E_{2},\ldots,E_{(q+1)^{2}}$ are linearly independent as elements of $M_{q+1}(\mathbb{C})$.

Suppose there exist $k_{\psi}\in\mathbb{C}$ and $\ell_{i}\in\mathbb{C}$ such that
\begin{equation}\label{eq3.2}
\sum_{\psi\in\widehat{N}\backslash \{1\}}k_{\psi}E_{\psi}+\sum_{i=1}^{q+1}\ell_{i}U_{i}=0,
\end{equation}
where $E_{\psi}=\frac{1}{q+1}|u_{\psi}\rangle\langle u_{\psi}|$ for each $\psi\in\widehat{N}\backslash \{1\}$,
and $U_{i}=\frac{1}{q+1}|e_{i}\rangle\langle e_{i}|$ for $i=1,2,\ldots,q+1$.

Comparing the $(i,i)$-th entry of Eq. (\ref{eq3.2}), where $i=1,2,\ldots,q+1$, we get that
\begin{equation}\label{eq3.3}
\frac{1}{(q+1)^{2}}\sum_{\psi\in\widehat{N}\backslash \{1\}}k_{\psi}+\frac{\ell_{i}}{q+1}=0.
\end{equation}

Comparing the $(i,j)$-th entry of Eq. (\ref{eq3.2}), where $1\leq i\neq j\leq q+1$, we have that
\begin{equation}\label{eq3.4}
\sum_{\psi\in\widehat{N}\backslash \{1\}}k_{\psi}\psi(d_{i}d_{j}^{-1})=0.
\end{equation}

Then, it follows from Lemma \ref{lemma3.3} that Eq. (\ref{eq3.4}) implies $k_{\psi}\equiv0$ for all $\psi\in\widehat{N}\backslash \{1\}$.
Inserting this into Eq. (\ref{eq3.3}), we have $\ell_{i}=0$ for $i=1,2,\ldots,q+1$.
Hence, $E_{1},E_{2},\ldots,E_{(q+1)^{2}}$ are linearly independent as elements of $M_{q+1}(\mathbb{C})$,
and so are $F_{1},F_{2},\ldots,F_{(q+1)^{2}}$. Consequently, $\mathcal{F}$ satisfies the informational completeness condition of
an $\varepsilon_{q+1}$-ASIC-POVM.

Therefore, we complete the proof. $\hfill\square$
\end{proof}

\subsection{Asymptotically optimal codebooks from Theorem \ref{theorem3.5}}\label{subsection 3.3}

An $(N,K)$ codebook $\mathcal{C}$ is a set $\{\mathbf{c}_{0},\mathbf{c}_{1},\ldots,\mathbf{c}_{N-1}\}$,
where $\mathbf{c}_{i}$ is a unit-norm complex vector in $\mathbb{C}^{K}$ for each $i=0,1,\ldots,N-1$.
The maximum cross-correlation amplitude, which is a performance measure of a codebook in practical applications, of the $(N,K)$ codebook $\mathcal{C}$ is defined by
\begin{equation*}
I_{\mathrm{max}}(\mathcal{C})=\max \limits_{0\leq j\neq k\leq N-1}|\mathbf{c}_{j}\mathbf{c}_{k}^{H}|,
\end{equation*}
where $\mathbf{c}_{k}^{H}$ denotes the conjugate transpose of the complex vector $\mathbf{c}_{k}$.

To evaluate an $(N,K)$ codebook $\mathcal{C}$, it is important to find the minimum achievable $I_{\mathrm{max}}(\mathcal{C})$.
Welch \cite{Welch1974Lower} showed that for any $(N,K)$ codebook $\mathcal{C}$ with $N\geq K$,
\begin{equation*}
I_{\mathrm{max}}(\mathcal{C})\geq I_{W}=\sqrt{\frac{N-K}{(N-1)K}},
\end{equation*}
and the equality holds if and only if $|\mathbf{c}_{j}\mathbf{c}_{k}^{H}|=\sqrt{\frac{N-K}{(N-1)K}}$ for all $0\leq j\neq k\leq N-1$.

If $I_{\mathrm{max}}(\mathcal{C})$ asymptotically meet the Welch bound $I_{W}$ for sufficiently large $K$,
i.e., $\lim \limits_{K\rightarrow\infty}\frac{I_{\mathrm{max}}(\mathcal{C})}{I_{W}}=1$, we call $\mathcal{C}$ an asymptotically optimal codebook.
Currently, constructing asymptotically optimal codebooks is a hot topic in coding theory (see, e.g., \cite{Luo2018Two,Heng2017New,Qian2023Gaussian}).

One can see that the vectors set $\mathcal{C}$ defined in Theorem \ref{theorem3.5} forms a $\left((q+1)^{2},q+1\right)$ codebook
with maximum cross-correlation amplitude $I_{\mathrm{max}}(\mathcal{C})=\frac{1}{\sqrt{q+1}}$. Further, we obtain that
\begin{equation*}
\lim \limits_{q\rightarrow\infty}\frac{I_{\mathrm{max}}(\mathcal{C})}{I_{W}}
=\lim \limits_{q\rightarrow\infty}\Big(\frac{1}{\sqrt{q+1}}\Big/\frac{1}{\sqrt{q+2}}\Big)=1,
\end{equation*}
which reveals that $\mathcal{C}$ is an asymptotically optimal codebook.

\subsection{Comparison of our results with known ones}

Theorem \ref{theorem3.5} and \cite[Theorem 3.2]{Luo2018New} construct $\varepsilon_{q+1}$-ASIC-POVMs.
Although their parameters are the same, we summarize other differences of these results in the following.
\vspace{-3pt}
\begin{itemize}
\item [(1)] To guarantee that the construction satisfies the approximate symmetry condition,
\cite[Theorem 3.2]{Luo2018New} utilizes the character sum $|\sum_{d\in D}\chi(d)|=\sqrt{q}$,
while Theorem \ref{theorem3.5} utilizes the Li bound $\big|\sum_{s\in S}\psi(s)\big|\leq\sqrt{q}$.
To guarantee that the construction satisfies the informational completeness condition,
\cite[Theorem 3.2]{Luo2018New} utilizes the Fourier transformation,
while Theorem \ref{theorem3.5} utilizes Lemma \ref{lemma3.3}.

\vspace{-3pt}

\item [(2)] \underline{Potential application}. For other possible applications of $\varepsilon_{q+1}$-ASIC-POVMs in the future (e.g., they might be useful in quantum cryptography \cite{Boileau2005Unconditional}, dynamic quantum tomography \cite{Czerwinski2021Quantum}, etc), we characterize ``how close'' the $\varepsilon_{q+1}$-ASIC-POVMs
are from being SIC-POVMs of dimension $q+1$. Specifically,
i) we give the explicit expressions about all $\varepsilon_{q+1}$ in different kinds of upper bounds of $(q+1)^{2}\mathrm{Tr}(F_{i}F_{j})$ in Theorem \ref{theorem3.5};
ii) we calculate the absolute value of the difference $|\frac{1}{q+2}-A|$ in Theorem \ref{theorem3.5}, where
$\frac{1}{q+2}$ is the value of $d^{2}\mathrm{Tr}(E_{i}E_{j})$ in Definition \ref{definition1.3} (ii) for $d=q+1$, and
$A$ is the explicit upper bound of $d^{2}\mathrm{Tr}(A_{i}A_{j})$ in Definition \ref{definition1.4} (ii) for $d=q+1$.

\vspace{-3pt}

\item [(3)] Two difference sets $D_{1}$ in group $G_{1}$ and $D_{2}$ in group $G_{2}$ are called equivalent if there exists a group isomorphism $\varrho$
between $G_{1}$ and $G_{2}$ such that $\varrho(D_{1})=\{\varrho(d):d\in D_{1}\}=gD_{2}$ for some $g\in G_{2}$.
As far as we know, although the set $S$ defined in Theorem \ref{theorem3.5} is a difference set with the same parameters as the set $D$ in
\cite[Theorem 3.2]{Luo2018New}, there is no evidence showing that these two difference sets are equivalent.

\vspace{-3pt}

\item [(4)] \underline{Practical application}. As shown in subsection \ref{subsection 3.3}, the set $\mathcal{C}$ corresponding to the $\varepsilon_{q+1}$-ASIC-POVM in
Theorem \ref{theorem3.5} forms a $\left((q+1)^{2},q+1\right)$ asymptotically optimal codebook.
The construction of asymptotically optimal codebooks is currently a hot topic in coding theory (see, e.g., \cite{Luo2018Two,Heng2017New,Qian2023Gaussian}).
\end{itemize}

\section{Concluding remarks}\label{section4}

In this paper, we constructed a class of $\varepsilon_{d}$-ASIC-POVMs in dimension $d=q$ and a class of $\varepsilon_{d}$-ASIC-POVMs in dimension $d=q+1$, respectively, where $q$ is a prime power.
In the first construction, we proved that all $2$-to-$1$ PN functions can be used for constructing $\varepsilon_{q}$-ASIC-POVMs.
We showed that there exist $2$-to-$1$ PN functions that do not satisfy the condition in \cite[Theorem III.3]{Cao2017Two}.
Therefore, the class of $\varepsilon_{q}$-ASIC-POVMs constructed in this paper is more general than the class constructed in \cite{Cao2017Two}.
We also showed that the set of vectors corresponding to the $\varepsilon_{q}$-ASIC-POVM forms a biangular frame.
The construction of $\varepsilon_{q+1}$-ASIC-POVMs is based on a multiplicative character sum estimate called the Li bound.
We showed that the set of vectors corresponding to the $\varepsilon_{q+1}$-ASIC-POVM forms a $\left((q+1)^{2},q+1\right)$ asymptotically optimal codebook.
We also characterized ``how close'' the $\varepsilon_{q}$-ASIC-POVMs (resp. $\varepsilon_{q+1}$-ASIC-POVMs) are from being SIC-POVMs
of dimension $q$ (resp. dimension $q+1$).

An interesting problem for the future is to construct $\varepsilon_{d}$-ASIC-POVMs via other character sum estimates related to some special kinds of functions
such as $m$-to-$1$ ($m\geq3$) mappings and rational functions over finite fields.
Besides, it would be interesting to construct IC-POVMs $\{\frac{1}{d}|u_{i}\rangle\langle u_{i}|:i=1,2,\ldots,d^{2}\}$
with the approximate symmetry condition $|\langle u_{i}|u_{j}\rangle|^{2}\leq \frac{c+\varepsilon_{d}}{d}$ for certain constant $c$ and all $1\leq i\neq j\leq d^{2}$, which are not restricted to the approximate symmetry condition $|\langle u_{i}|u_{j}\rangle|^{2}\leq \frac{1+\varepsilon_{d}}{d}$ of $\varepsilon_{d}$-ASIC-POVMs.

Finally, let us explain the significance of constructing $\varepsilon_{d}$-ASIC-POVMs from the following two points.
\vspace{-3pt}
\begin{itemize}
\item [(1)] As we know, SIC-POVMs are a special type of IC-POVMs, which is important for quantum state tomography (see, e.g., \cite{Caves2002Unknown,Scott2006Tight}).
However, the theoretical construction of SIC-POVMs poses challenges.
Affected by noise, it is also difficult to accurately prepare SIC-POVMs in experiments.
Some SIC-POVMs with small dimensions, such as dimensions 2 and 3, have been prepared in experiments, see \cite{Bian2015Realization,Wang2023Generalized}.
In general, noise is difficult to avoid, and even when the noise is relatively small, the SIC-POVM prepared in the experiment is often an approximate form,
which can be regarded as an $\varepsilon_{d}$-ASIC-POVM.
Therefore, in terms of existence, currently the $\varepsilon_{d}$-ASIC-POVMs are more extensive than SIC-POVMs both in theoretical construction and experimental preparation.

\vspace{-3pt}

\item [(2)] Since $\varepsilon_{d}$-ASIC-POVMs are IC-POVMs having the property of informational completeness, they can be used in classical shadow tomography.
To be specific, instead of obtaining the full information of the unknown state $\rho$, people aim to predict its properties based on the values of $\{\mbox{tr}(\rho M_k)\}$ (see, e.g., \cite{Huang2020Predicting}), where $\{M_k\}$ are the observables related to the properties.
The idea promotes quantum advantage in learning from experiments \cite{Huang2022Quantum}. It is proved that the sample complexity could decrease exponentially if we randomly choose the IC measurements and use the classical shadow tomography method. For different properties, different IC measurements could produce different predicting accuracy. For example, the Pauli measurements are not good at predicting the entanglement property.
In \cite{Nguyen2022Optimizing}, the authors formulated classical shadow tomography by changing projective measurements to generalized POVMs.
Based on the above literature, we believe that $\varepsilon_{d}$-ASIC-POVMs can be used in classical shadow tomography.
Moreover, when the approximately symmetric property is taken into consideration, namely, the usage of $\varepsilon_{d}$-ASIC-POVMs might bring some advantage in predicting some properties,
e.g., reducing sampling complexity, expanding the types of observables that can be used for prediction, etc.
\end{itemize}

\section*{Acknowledgment}

The authors would like to sincerely thank the two anonymous referees for their detailed and constructive comments and suggestions that greatly improved the quality of this article. Meng Cao is grateful to Yu Wang for useful discussions on classical shadow tomography.
The research of Xiantao Deng was supported in part by National Natural Science Foundation of China (Grant No. 12171331).

\section*{Data availability statement}

Data sharing is not applicable to this paper as no datasets were generated or analyzed during the current study.

\section*{Conflict of interest statement}

The authors declare that they have no known competing financial interests or personal relationships that could have appeared to
influence the work reported in this paper.

\section*{Appendix: Cardinality of all $2$-to-$1$ mappings over $\mathbb{F}_{q}$}
\renewcommand{\thetheorem}{\Alph{theorem}}

The following proposition gives the cardinality of all $2$-to-$1$ mappings over $\mathbb{F}_{q}$.

\begin{proposition}\label{propositionE}
Let $q=p^{k}$ be a prime power. Let $L$ denote the set consisting of all $2$-to-$1$ mappings over $\mathbb{F}_{q}$.
Then,

\vspace{4pt}

(1) if $p=2$, we have $\sharp L=\displaystyle\frac{(2^{k}!)^{2}}{2^{2^{k-1}}(2^{k-1}!)^{2}}$;

\vspace{4pt}

(2) if $p\neq2$, we have $\sharp L=\displaystyle\frac{p^{2k}\big((p^{k}-1)!\big)^{2}}{2^{\frac{p^{k}-1}{2}}\big(\frac{p^{k}-1}{2}!\big)^{2}}$.
\end{proposition}

\begin{proof}
(1): See \cite[Proposition 3]{Mesnager2019On}.

(2): For $p\neq2$, let $f\in L$ be a $2$-to-$1$ mapping over $\mathbb{F}_{p^{k}}$. Then,
\begin{itemize}
\item In $\mathbb{F}_{p^{k}}$, there is only one element, denoted by $c_{1}$, has exactly one preimage of $f$;

\item In $\mathbb{F}_{p^{k}}\backslash\{c_{1}\}$, there are $\frac{p^{k}-1}{2}$ elements, denoted by $c_{2},c_{3},\ldots,c_{\frac{p^{k}+1}{2}}$,
have two preimages of $f$, while the remaining $\frac{p^{k}-1}{2}$ elements have zero preimage of $f$.
\end{itemize}

For $c_{1}$, it has $p^{k}$ choices in $\mathbb{F}_{p^{k}}$, and its preimage $f^{-1}(c_{1})$ has also $p^{k}$ choices in $\mathbb{F}_{p^{k}}$.

For $c_{2},c_{3},\ldots,c_{\frac{p^{k}+1}{2}}$, they have $\dbinom{p^{k}-1}{\frac{p^{k}-1}{2}}$ choices in $\mathbb{F}_{p^{k}}\backslash\{c_{1}\}$.
Further, for $c_{2}$, its preimage $f^{-1}(c_{2})$ has $\dbinom{p^{k}-1}{2}$ choices;
for $c_{3}$, its preimage $f^{-1}(c_{3})$ has $\dbinom{p^{k}-3}{2}$ choices;
finally, for $c_{\frac{p^{k}+1}{2}}$, its preimage $f^{-1}\big(c_{\frac{p^{k}+1}{2}}\big)$ has $\dbinom{2}{2}$ choices.

\vspace{6pt}

Therefore, we calculate that
\begin{align*}
\sharp L&=p^{k}\cdot p^{k}\cdot\dbinom{p^{k}-1}{\frac{p^{k}-1}{2}}\cdot\dbinom{p^{k}-1}{2}\cdot\dbinom{p^{k}-3}{2}\cdots\dbinom{4}{2}\dbinom{2}{2}\nonumber\\
&=p^{2k}\cdot\frac{(p^{k}-1)!}{(\frac{p^{k}-1}{2}!)^{2}}\cdot\frac{(p^{k}-1)!}{(p^{k}-3)!\cdot2!}\cdot\frac{(p^{k}-3)!}{(p^{k}-5)!\cdot2!}
\cdots\frac{6!}{4!\cdot2!}\cdot\frac{4!}{2!\cdot2!}\nonumber\\
&=\frac{p^{2k}\big((p^{k}-1)!\big)^{2}}{2^{\frac{p^{k}-1}{2}}\big(\frac{p^{k}-1}{2}!\big)^{2}},
\end{align*}
which completes the proof. $\hfill\square$
\end{proof}

\footnotesize{
\bibliographystyle{plain}
\bibliography{revised}
}

\end{document}